\documentclass[submission,copyright,creativecommons]{eptcs}
\providecommand{\event}{LSFA 2024} 

\usepackage{iftex}
\usepackage{amsthm}
\newtheorem{theorem}{Theorem}
\newtheorem{definition}{Definition}
\newtheorem{example}{Example}
\newtheorem{lemma}{Lemma}
\newtheorem{remark}{Remark}
\newtheorem{corollary}{Corollary}

\usepackage{latexsym,relsize}
\usepackage{amsmath,amsfonts,amssymb}
\usepackage{mathtools}
\usepackage{algorithm}
\usepackage{algpseudocode}
\usepackage{setspace}
\usepackage{multicol}
\usepackage{longtable}
\usepackage{color}
\usepackage{bussproofs}
\EnableBpAbbreviations

\usepackage{tikz}
\usepackage{colonequals}

\usepackage{multirow}

\usepackage{multicol}

\usepackage{array}
\newcolumntype{C}[1]{>{\centering\arraybackslash}p{#1}}
\newcolumntype{L}[1]{>{\arraybackslash}p{#1}}

\def\mc{\multicolumn}

\newcommand{\fns}{\footnotesize}

\newcommand{\val}[1]{[\![{#1}]\!]}
\newcommand{\descr}[1]{(\![{#1}]\!)}

\newcommand{\Diamondblack}{\blacklozenge}

\newcommand{\fakeparagraph}[1]{

\textit{#1} \ \ }

\newcommand{\ceq}{\colonequals}

\ifpdf
  \usepackage{underscore}         
  \usepackage[T1]{fontenc}        
\else
  \usepackage{breakurl}           
\fi

\renewcommand{\phi}{\varphi}

\title{Fuzzy Lattice-based Description Logic}
\author{Yiwen Ding\qquad\qquad Krishna Manoorkar
\institute{School of Business and Economics\\
Vrije Universiteit Amsterdam\thanks{This project has received funding from the European Union’s Horizon 2020 research and innovation programme under the Marie Skłodowska-Curie grant agreement No 101007627. The research of Krishna Manoorkar is supported by the NWO grant KIVI.2019.001. Yiwen Ding is supported by the China Scholarship Council.}\\
Amsterdam, Netherlands}
\email{dyiwen666@gmail.com \quad\qquad krishna.manoorkar@gmail.com}
}
\def\titlerunning{Fuzzy Lattice-based Description Logic}
\def\authorrunning{Y. Ding, K. Manoorkar}
\begin{document}
\maketitle

\begin{abstract}
Recently, description logic $\mathrm{LE}$-$\mathcal{ALC}$ was introduced for reasoning in the semantic environment of enriched formal contexts, and a polynomial-time tableaux algorithm was developed to check the consistency of knowledge bases with acyclic TBoxes \cite{van2023non}. In this work, we introduce a fuzzy generalization of $\mathrm{LE}$-$\mathcal{ALC}$ called  $\mathrm{LE}$-$\mathcal{FALC}$ which provides description logic counterpart of many-valued normal non-distributive logic a.k.a.~many-valued $\mathrm{LE}$-logic. This description logic can be used to represent and reason about knowledge in formal framework  of fuzzy formal contexts and fuzzy formal concepts as introduced in \cite{Belohlavek1999-BLOFGC}. We provide a tableaux algorithm that provides a complete and sound polynomial-time decision procedure to check the consistency of $\mathrm{LE}$-$\mathcal{FALC}$ ABoxes. As a result, we also obtain an exponential-time decision procedure for checking the consistency of $\mathrm{LE}$-$\mathcal{FALC}$ with acyclic TBoxes by unraveling. 
\end{abstract}

\section{Introduction}

Description Logic (DL) \cite{DLhandbook} is a class of logical formalisms, typically based on the classical first-order logic, widely used in Knowledge Representation and Reasoning to describe and reason about relevant concepts and their relationships in a given application domain.

 {\em Normal non-distributive modal logic} a.k.a.~$\mathrm{LE}$-logic has been studied as a logic of categorization endowed with modal operators based on  the fact that non-distributive modal logic is sound and complete  w.r.t.~its semantics based on  {\em enriched formal contexts} (i.e.~, relational structures based on formal contexts from Formal Concept Analysis (FCA)) \cite{conradie2016categories,conradie2017toward}. An enriched formal context dually corresponds to the concept lattice of its underlying formal context expanded with normal modal operators defined by those enriching relations. Similarly to the classical modal logic, these normal modal operators have many different intuitive interpretations, such as epistemic interpretation \cite{conradie2017toward}, approximation interpretation \cite{conradie2021rough}, etc.

In \cite{van2023non}, the two-sorted non-distributive description logic $\mathrm{LE}$-$\mathcal{ALC}$\footnote{Even though concept names in $\mathrm{LE}$-$\mathcal{ALC}$ do  not contain negation, we still refer to this description logic as $\mathrm{LE}$-$\mathcal{ALC}$ rather than $\mathrm{LE}$-$\mathcal{ALE}$, as negation on ABox terms is included in the description logic language.} was developed, based on  $\mathrm{LE}$-logics and their semantics based on enriched formal contexts. $\mathrm{LE}$-$\mathcal{ALC}$ provides a natural means of reasoning about the formal concepts (or categories) arising from formal contexts in FCA \cite{ganter1997applied,ganter2012formal} enriched with modal operators. $\mathrm{LE}$-$\mathcal{ALC}$ has the same  relationship with non-distributive modal logic and its semantics based on formal contexts, as $\mathcal{ALC}$ with the classical normal modal logic and its Kripke frames semantics. Namely, $\mathrm{LE}$-$\mathcal{ALC}$  facilitates the description of {\em enriched formal contexts}, and {\em formal concepts} generated by them. As enriched formal contexts dually correspond to complete lattices equipped with normal modal operators, $\mathrm{LE}$-$\mathcal{ALC}$ is also a natural framework for representing and reasoning about general (possibly non-distributive) complete lattices equipped with  normal modal operators.

For many real-life applications, the concepts under consideration are {\em imprecise} or {\em fuzzy}. For example, if we want to categorize movies, the concepts involved such as ``action movies", ``dramas", etc.~are imprecise. To model such scenarios, different fuzzy generalizations of FCA have been studied extensively \cite{Belohlavek1999-BLOFGC,fuzzyandroughFCA,poelmans2013formal}. 
In \cite{conradie2019logic} the many-valued $\mathrm{LE}$-logic was described as the logic of vague categorization. The semantics of this logic is given by many-valued formal contexts enriched with relations defining modal operators.  The propositional part of this logic corresponds to the logic of fuzzy concepts defined in \cite{Belohlavek1999-BLOFGC}. The modal operators can be given different interpretations, such as epistemic interpretation or approximation interpretation as in the crisp case \cite{conradie2019logic}.  In this work, we generalize the description logic $\mathrm{LE}$-$\mathcal{ALC}$ to the fuzzy description logic $\mathrm{LE}$-$\mathcal{FALC}$. This logic has the same relationship with the many-valued $\mathrm{LE}$-logic as $\mathrm{LE}$-$\mathcal{ALC}$ with (crisp) $\mathrm{LE}$-logic. $\mathrm{LE}$-$\mathcal{FALC}$ facilitates the description of many-valued enriched formal contexts, which give rise to fuzzy concept lattices extended with normal modal operators. Hence, $\mathrm{LE}$-$\mathcal{FALC}$ is also a natural framework to represent and reason about general fuzzy FCA expanded with modal operators.

In \cite{van2023non}, a polynomial-time tableaux algorithm   for checking the  consistency of $\mathrm{LE}$-$\mathcal{ALC}$ ABoxes was developed. The polynomial-time upper bound is based on the following two features of $\mathrm{LE}$-$\mathcal{ALC}$: (1) In the semantics of  $\mathrm{LE}$-$\mathcal{ALC}$ disjunction is interpreted in terms of the {\em intersection} of intensions. Therefore, $\vee$ does not produce branching for the same reason that $\wedge$ does not in the classical setting.  (2) Unlike the classical logic, the $\Diamond$ operator in $\mathrm{LE}$-$\mathcal{ALC}$ is interpreted using a universal quantifier (instead of an existential quantifier), which allows us to bound the number of new constants appearing in the tableaux expansion. 

In this work, we generalize the tableaux expansion rules for $\mathrm{LE}$-$\mathcal{ALC}$ to the fuzzy setting to define tableaux expansion rules for $\mathrm{LE}$-$\mathcal{FALC}$. We show that the resulting tableaux algorithm provides a sound and complete polynomial-time decision procedure for checking the consistency of $\mathrm{LE}$-$\mathcal{FALC}$ knowledge bases with acyclic TBoxes. Since many description logic reasoning tasks can be equivalently represented as a problem of  checking the consistency of knowledge bases (cf.~\cite{DLhandbook}), this tableaux algorithm  provides us with an efficient methodology to perform these tasks for $\mathrm{LE}$-$\mathcal{FALC}$ in polynomial time. 

\paragraph{Structure of the paper.} In Section \ref{sec:Preliminaries}, we present the required preliminaries used in this paper. In Section \ref{sec:LE-FALC}, we introduce the description logic $\mathrm{LE}$-$\mathcal{FALC}$. In Section \ref{Sec: tableau}, we introduce the Tableaux algorithm for checking the consistency of $\mathrm{LE}$-$\mathcal{FALC}$ ABoxes, which can also be easily extended to $\mathrm{LE}$-$\mathcal{FALC}$ acyclic TBoxes. In Section \ref{sec:soundness}, we show the soundness of this tableaux algorithm. In Section \ref{sec:Completeness}, we show the completeness of this tableaux algorithm. In Section \ref{sec:example of knowledge bases}, we give some examples of $\mathrm{LE}$-$\mathcal{FALC}$ ABoxes and use our tableaux algorithm to check their consistency.

\section{Preliminaries}\label{sec:Preliminaries}

In this section, we gather some useful facts about the  many-valued $\mathrm{LE}$-logic and its many-valued polarity-based semantics \cite{conradie2019logic}, and description logic $\mathrm{LE}$-$\mathcal{ALC}$ \cite{van2023non}. We assume that the readers are familiar with the basic concepts from description logic, which we refer to \cite{baader2017introduction}. For more details on $\mathrm{LE}$-logic and its polarity-based semantics, we refer to \cite{conradie2017toward}, \cite{conradie2021rough},  and \cite{conradie2020non}.

\subsection{Many-valued polarity-based semantics}\label{subsec:Many-valued polarity-based semantics}


Let $\mathrm{Prop}$ be a countable set of propositional variables. The language of many-valued $\mathrm{LE}$-logic $\mathcal{L}$  is defined as follows:

{{\centering
  $\varphi ::= p \mid  \bot \mid \top \mid \varphi \wedge \varphi \mid \varphi \vee \varphi \mid \Box \varphi \mid \Diamond \varphi$,  
\par}}
\noindent
where $p\in \mathrm{Prop}$, and $\Box\in\mathcal{G}$ and $\Diamond\in\mathcal{F}$ for finite sets $\mathcal{G}$ and $\mathcal{F}$ of unary $\Box$-type (resp.~$\Diamond$-type) modal operators. 

In this paper, we let $\mathbf{H} = (H, \vee, \wedge, \rightarrow, 1, 0)$ be a complete and completely distributive Heyting algebra.\footnote{{We chose Heyting algebra here in order to maintain the simplicity and readability of the article. In fact, the results of this paper can be generalized to any complete frame-distributive and dually frame-distributive, commutative, and associative residuated lattice.(cf.~\cite{conradie2019logic})}}For any $\alpha,\beta$ in $\mathbf{H}$, let $\alpha\leftrightarrow\beta:=\alpha\rightarrow\beta\wedge\beta\rightarrow\alpha$. For any non-empty set $W$, an {\em $\mathbf{H}$-subset} of $W$ is a map $u: W\to \mathbf{H}$. We let $\mathbf{H}^W$ denote the set of all $\mathbf{H}$-subsets. Clearly, $\mathbf{H}^W$ induces a complete and completely distributive Heyting algebra by defining the operations pointwise. 
For every $\alpha\in \mathbf{H}$ and $w\in W$, $\{\alpha/ w\}: W\to \mathbf{H}$ is the map defined by $w'\mapsto \alpha$ if $w' = w$ and $w'\mapsto 0$ if $w'\neq w$. 
Let $u, v: W\to \mathbf{H}$ be any $\mathbf{H}$-subsets, we define $u\subseteq v$ if $u(w)\leq v(w)$ for each $w\in W$, which defines a partial order on $\mathbf{H}^W$. An {\em $\mathbf{H}$-relation} is a map $R: U \times W \rightarrow \mathbf{H}$, where $U$ and $W$ are non-empty sets. Any $\mathbf{H}$-relation $R: U \times W \rightarrow \mathbf{H}$ induces  maps $R^{(0)}[-] : \mathbf{H}^W \rightarrow \mathbf{H}^U$ and $R^{(1)}[-] : \mathbf{H}^U \rightarrow \mathbf{H}^W$ which are defined as follows: for every $f: U \to \mathbf{H}$, and $u: W \to \mathbf{H}$, $R^{(1)}[f]:W\to \mathbf{H}$ such that $x\mapsto \bigwedge_{a\in U}(f(a)\rightarrow R(a, x))$, and $R^{(0)}[u]: U\to \mathbf{H}$ such that $a\mapsto \bigwedge_{x\in W}(u(x)\rightarrow R(a, x))$.

\begin{definition}\label{def: A-polarity}
An $\mathbf{H}$-{\em valued formal context} is a tuple $\mathfrak{P} = (A, X, I)$ such that $A$ and $X$ are non-empty sets, and $I: A\times X\to \mathbf{H}$ is an $\mathbf{H}$-relation. Any $\mathbf{H}$-valued formal context induces maps $(\cdot)^{\uparrow}: \mathbf{H}^A\to \mathbf{H}^X$ and $(\cdot)^{\downarrow}: \mathbf{H}^X\to \mathbf{H}^A$ given by $(\cdot)^{\uparrow} = I^{(1)}[\cdot]$ and $(\cdot)^{\downarrow} = I^{(0)}[\cdot]$. 
\end{definition}

Given an $\mathbf{H}$-valued formal context $(A, X, I)$, it is easy to see that for any $f\in \mathbf{H}^A$ and $u\in \mathbf{H}^X$, there is $f\subseteq u^{\downarrow}$ iff $u\subseteq  f^{\uparrow}$, which means the pair of maps $(\cdot)^{\uparrow}$ and $(\cdot)^{\downarrow}$ form a Galois connection between $(\mathbf{H}^A, \subseteq)$ and $(\mathbf{H}^X, \subseteq)$. A {\em fuzzy formal concept} of $\mathfrak{P}$ is a pair of maps $(f, u)\in \mathbf{H}^A\times \mathbf{H}^X$ such that $f^{\uparrow}=u$ and $u^{\downarrow}=f$. It follows immediately that if a pair of maps $(f, u)\in\mathbf{H}^A\times \mathbf{H}^X$ is a fuzzy formal concept, then there is $f^{\uparrow \downarrow} = f$ and $u^{\downarrow\uparrow} = u$, which means that $f$ and $u$ are {\em Galois-stable}. In this paper, we may use pair $(\val{c}, \descr{c})$ to denote a fuzzy concept $c$ from a given fuzzy formal context where $\val{c}\in\mathbf{H}^A$ and $\descr{c}\in\mathbf{H}^X$, except Section \ref{subsec: LE-ALC} and Section \ref{app:LE-ALC}. $\val{c}$ (resp.~$\descr{c}$) is the {\em extension} (resp.~{\em intension}) of the fuzzy concept. The set of all fuzzy formal concepts of $\mathfrak{P}$ can be partially ordered as follows:

{{\centering
$(f, u)\leq (g, v)\quad \mbox{ iff }\quad f\subseteq g \quad \mbox{ iff }\quad v\subseteq u$.
\par}}
\noindent
Ordered in this way, the set of fuzzy formal concepts of $\mathfrak{P}$ forms a complete lattice, which we refer to as the {\em concept lattice} of $\mathfrak{P}$ and denote by $\mathfrak{P}^+$, such that for  set  $\mathrm{K}$ of fuzzy formal concepts  of $\mathfrak{P}$, $\bigwedge\mathrm{K} := (\bigwedge_{c\in\mathrm{K}}\val{c},(\bigwedge_{c\in\mathrm{K}}\val{c})^\uparrow)$ and $\bigvee\mathrm{K} := (\bigwedge_{c\in\mathrm{K}}\descr{c}))^\downarrow,\bigwedge_{c\in\mathrm{K}}\descr{c})$.

\begin{definition}\label{def: A-frame}
An {\em $\mathbf{H}$-valued enriched formal context} is a tuple $\mathfrak{F} = (\mathfrak{P}, \mathcal{R}_\Box, \mathcal{R}_\Diamond)$ such that $\mathfrak{P} = (A, X, I)$ is an $\mathbf{H}$-valued formal context, and $\mathcal{R}_\Box=\{R_\Box: A\times X\to \mathbf{H}\mid\Box\in\mathcal{G}\}$ and $\mathcal{R}_\Diamond=\{R_\Diamond: X\times A\to \mathbf{H}\mid \Diamond\in\mathcal{F}\}$ are sets of $I$-{\em compatible} $\mathbf{H}$-relations, that is, for any $\Box\in\mathcal{G}$ and $\Diamond\in\mathcal{F}$, $a \in A$ and $x \in X$, $\alpha \in \mathbf{H}$, the $\mathbf{H}$-subsets $R_{\Box}^{(0)}[\{\alpha / x\}]$, $R_{\Box}^{(1)}[\{\alpha / a\}]$,  $R_{\Diamond}^{(0)}[\{\alpha / a\}]$ and $R_{\Diamond}^{(1)}[\{\alpha / x\}]$ are Galois-stable.   
\end{definition}

Every $\mathbf{H}$-valued formal context can be seen as a many-value formal context, where elements in $A$ are "objects" and elements in $X$ are "features". For any $a\in A$, $x\in X$ and $\alpha\in\mathbf{H}$, $I(a,x)=\alpha$ is read as "object $a$ has feature $x$  to degree $\alpha$". The $I$-compatibility conditions can be understood in such way: the sets of all objects (resp.~features) relating to a feature (resp.~object) by modal relations $R_\Box$ and $R_\Diamond$ form concepts. These modal relations have different interpretations like epistemic interpretation \cite{conradie2017toward} or approximation interpretation \cite{conradie2021rough}.

\begin{definition}
    Let $\mathfrak{F}=(\mathfrak{P}, \mathcal{R}_\Box, \mathcal{R}_\Diamond)$ be any $\mathbf{H}$-valued enriched formal context. $\mathfrak{F}^+=(\mathfrak{P}^+, \{[R_\Box]\}_{\Box\in\mathcal{G}}, \{\langle R_\Diamond\rangle\}_{\Diamond\in\mathcal{F}})$ is the {\em complex algebra} of $\mathfrak{F}$, where $\mathfrak{P}^+$ is the concept lattice of $\mathfrak{P}$, and for any $\Box\in\mathcal{G}$ and $\Diamond\in\mathcal{F}$, $[R_{\Box}], \langle R_{\Diamond}\rangle : \mathfrak{P}^{+} \to \mathfrak{P}^{+}$ are maps such that for every $c=(\val{c}, \descr{c}) \in \mathfrak{P}^{+}$, $[R_{\Box}](c):=(R_{\Box}^{(0)}[\descr{c}], (R_{\Box}^{(0)}[\descr{c}])^{\uparrow})$ and $\langle R_{\Diamond}\rangle(c):=  ((R_{\Diamond}^{(0)}[\val{c}])^{\downarrow}, R_{\Diamond}^{(0)}[\val{c}])$.
\end{definition}

An example of $\mathbf{H}$-valued enriched formal context and its complex algebra is provided below.

\begin{example}
  Consider the Heyting algebra $\mathbf{H}$ which is a three-valued chain $(\{0,1/2, 1\}, \leq)$. Note that, on this algebra, the operation $\rightarrow $   is given by $a\rightarrow b=1$ if $a\leq b$, and $a\rightarrow b=b$ otherwise. Let $(A, X, I)$ be the fuzzy formal context where $A=\{a_1, a_2\}$, $X=\{x_1, x_2, x_3\}$, with incidence relation $I: A \times X \to \mathbf{H}$ enriched with $I$-{\em compatible} fuzzy relations $R_\Box:  A \times X \to \mathbf{H}$ and $R_\Diamond:  X \times A \to \mathbf{H}$   given in the following tables:
\begin{center}
\begin{tabular}{ |c|c|c|  } 
\hline
 I & $a_1$ & $a_2$ \\
\hline
 $x_1$ & $1$ & $1/2$\\
\hline
 $x_2$ & $0$ & $1$\\
 \hline
 $x_3$ & $1$ & $1/2$\\
 \hline
\end{tabular}
\quad 
\begin{tabular}{ |c|c|c|  } 
\hline
 \textbf{$R_\Box$} & $a_1$ & $a_2$ \\
\hline
 $x_1$ & $1$ & $1$\\
\hline
 $x_2$ & $0$ & $1$\\
 \hline
 $x_3$ & $1$ & $1$\\
 \hline
\end{tabular}
\quad 
\begin{tabular}{ |c|c|c|c|  } 
\hline
 \textbf{$R_\Diamond$} & $x_1$  & $x_2$  & $x_3$\\
\hline
 $a_1$ & $1$  & $0$  & $1$\\
\hline
 $a_2$ & $1$ & $1$   & $1$\\
 \hline
\end{tabular}
\end{center}
\smallskip

It is easy to check that the only  Galois-stable fuzzy subsets of $A$ (extents of the concepts generated) are the following:
\begin{center}
\begin{tabular}{ |c|c|c| } 
\hline
 & $a_1$ & $a_2$ \\
\hline
 $f_1$ & $1$ & $1$\\
\hline
 $f_2$ & $1$ & $1/2$\\
 \hline
 $f_3$ & $0$ & $1$\\
 \hline 
 $f_4$ & $0$ & $1/2$\\
 \hline
\end{tabular}
\end{center}

Therefore, the  fuzzy sets  $g_i = f_i^\uparrow$  (intents of the concepts generated) are given as follows:
\begin{center}
\begin{tabular}{ |c|c|c|c| } 
\hline
 & $x_1$ & $x_2$ & $x_3$ \\
\hline
 $g_1$ & $1/2$ & $0$ & $1/2$\\
\hline
 $g_2$ & $1$ & $0$ & $1$\\
 \hline
 $g_3$ & $1/2$ & $1$ & $1/2$\\
 \hline 
 $g_4$ & $1$ & $1$ & $1$\\
 \hline
\end{tabular}
\end{center}
The concept lattice of $(A, X, I)$  is represented by the Hasse diagram \ref{fig:Hasse-diagram}.
The operations $[R_\Box]$ and $\langle R_\Diamond \rangle$ on the lattice are given  as follows:

\begin{center}
    \begin{tabular}{|c|c|c|c|c|}
    \hline
         &  $(f_1,g_1)$ & $(f_2,g_2)$ & $(f_3,g_3)$ & $(f_4,g_4)$\\
         \hline
        $[R_\Box]$  & $(f_1,g_1)$  &  $(f_1,g_1)$ & $(f_3,g_3)$ & $(f_3, g_3)$ \\
        \hline
        $\langle R_\Diamond \rangle$ & $(f_2,g_2)$ & $(f_2,g_2)$ & $(f_4,g_4)$ & $(f_4,g_4)$\\
    \hline
    \end{tabular}
\end{center}
\end{example}

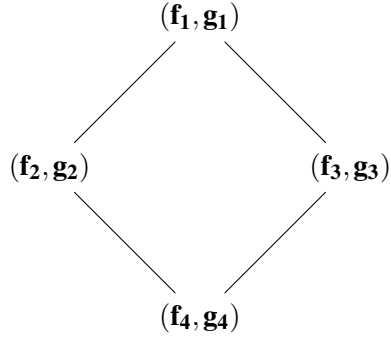
\begin{figure}
\begin{center}
\begin{tikzpicture}
    \node (A) at (0,2) {\(\mathbf{(f_1, g_1)}\)};
    \node (B) at (-2,0) {\(\mathbf{(f_2, g_2)}\)};
    \node (C) at (2,0) {\(\mathbf{(f_3, g_3)}\)};
    \node (D) at (0,-2) {\(\mathbf{(f_4, g_4)}\)};

    \draw (A) -- (B);
    \draw (A) -- (C);
    \draw (B) -- (D);
    \draw (C) -- (D);
\end{tikzpicture}
\end{center}
\caption{The concept lattice of $(A, X, I)$}
\label{fig:Hasse-diagram}
\end{figure}

\begin{definition}
 An $\mathbf{H}$-{\em model} is a tuple $\mathfrak{M} = (\mathfrak{F}, V)$ such that $\mathfrak{F}$ is an $\mathbf{H}$-valued enriched formal context, and $V: \mathsf{Prop}\to \mathfrak{F}^+$ is called an {\em assignment} on $\mathfrak{F}$. For any $p\in \mathsf{Prop}$, let $V(p): = (\val{p}, \descr{p})$, where $\val{p}: A\to \mathbf{H}$ and $\descr{p}: X\to\mathbf{H}$ such that $\val{p}^\uparrow = \descr{p}$ and $\descr{p}^\downarrow = \val{p}$. $V$ can be homomorphically extended to a unique {\em valuation} $\overline{V}:\mathcal{L}\rightarrow\mathfrak{F}^+$.
\end{definition}

\begin{definition}
Given any $\mathbf{H}$-model $\mathfrak{M}=(\mathfrak{F}, V)$ and $\alpha\in \mathbf{H}$, let $1^{\mathbf{H}^A}:A\rightarrow\mathbf{H}$ be the constant map such that $a\mapsto 1$ for any $a\in A$, and $1^{\mathbf{H}^X}:X\rightarrow\mathbf{H}$ be the constant map such that $x\mapsto 1$ for any $x\in X$. There are two {\em modal satisfaction relations} $\Vdash^\alpha$ and $\succ^\alpha$ defined inductively as follows:

\begin{center}
		\begin{tabular}{r c l}
			$\mathfrak{M}, a\Vdash^\alpha p$ & iff & $\alpha\leq \val{p}(a)$;\\
			$\mathfrak{M}, a\Vdash^\alpha \top$ & iff & $\alpha\leq (1^{\mathbf{H}^A})(a)$ i.e.~always;\\
			$\mathfrak{M}, a\Vdash^\alpha \bot$ & iff & $\alpha\leq (1^{\mathbf{H}^X})^\downarrow (a) = \bigwedge_{x\in X}(1^{\mathbf{H}^X}(x)\to I(a, x)) =  \bigwedge_{x\in X} I(a, x)$;\\
			$\mathfrak{M}, a\Vdash^\alpha \varphi\wedge \psi$ & iff & $\mathfrak{M}, a\Vdash^\alpha \varphi\quad $ and $\quad\mathfrak{M}, a\Vdash^\alpha \psi$;\\
			$\mathfrak{M}, a\Vdash^\alpha \varphi\vee \psi$ & iff & $\alpha\leq (\descr{\varphi}\wedge\descr{\psi})^\downarrow(a) = \bigwedge_{x\in X}(\descr{\varphi}(x)\wedge\descr{\psi}(x)\to I(a, x))$;\\
			$\mathfrak{M}, a\Vdash^\alpha \Box \varphi$ & iff & $\alpha\leq (R^{(0)}_\Box[\descr{\varphi}])(a) = \bigwedge_{x\in X}(\descr{\varphi}(x)\to R_\Box(a, x))$;\\
			$\mathfrak{M}, a\Vdash^\alpha \Diamond \varphi$ & iff & $\alpha\leq ((R^{(0)}_\Diamond[\val{\varphi}])^\downarrow)(a) = \bigwedge_{x\in X}((R^{(0)}_\Diamond[\val{\varphi}])(x)\to I(a, x))$\\
			$\mathfrak{M}, x\succ^\alpha p$ & iff & $\alpha\leq \descr{p}(x)$;\\
			$\mathfrak{M}, x\succ^\alpha \bot$ & iff & $\alpha\leq (1^{\mathbf{H}^X})(x)$ i.e.~always;\\
			$\mathfrak{M}, x\succ^\alpha \top$ & iff & $\alpha\leq (1^{\mathbf{H}^A})^\uparrow (x) = \bigwedge_{a\in A}(1^{\mathbf{H}^A}(a)\to I(a, x)) =  \bigwedge_{a\in A} I(a, x)$;\\
			$\mathfrak{M}, x\succ^\alpha \varphi\vee \psi$ & iff & $\mathfrak{M}, x\succ^\alpha \varphi\quad $ and $\quad\mathfrak{M}, x\succ^\alpha \psi$;\\
			$\mathfrak{M}, x\succ^\alpha \varphi\wedge \psi$ & iff & $\alpha\leq (\val{\varphi}\wedge\val{\psi})^\uparrow(x) = \bigwedge_{a\in A}(\val{\varphi}(a)\wedge\val{\psi}(a)\to I(a, x))$;\\
			$\mathfrak{M}, x\succ^\alpha \Diamond \varphi$ & iff & $\alpha\leq (R^{(0)}_\Diamond[\val{\varphi}])(x) = \bigwedge_{a\in A}(\val{\varphi}(a)\to R_\Diamond(x, a))$, for any $\Diamond\in\mathcal{F}$;\\
			$\mathfrak{M}, x\succ^\alpha \Box \varphi$ & iff & $\alpha\leq ((R^{(0)}_\Box[\descr{\varphi}])^\uparrow)(x) = \bigwedge_{a\in A}((R^{(0)}_\Box[\descr{\varphi}])(a)\to I(a, x))$, for any $\Box\in\mathcal{G}$.\\
		\end{tabular}
	\end{center}  
\end{definition}

{From the above definition,  it is easy to check that for every $\varphi\in \mathcal{L}$, there is $\mathfrak{M}, a\Vdash^\alpha \varphi\mbox{ iff }\alpha\leq \val{\varphi}(a)\mbox{, and }\mathfrak{M}, x\succ^\alpha \varphi \mbox{ iff }\alpha\leq \descr{\varphi}(x)$. Therefore, given any $\mathbf{H}$-model $\mathfrak{M}$, $a,x$ in $\mathfrak{M}$, $\alpha\in\mathbf{H}$ and $\varphi\in\mathcal{L}$, 
$\mathfrak{M}, a\Vdash^\alpha\varphi$ and $\mathfrak{M}, x\succ^\alpha\varphi$ can be read as "object $a$ is a member of category $\varphi$  to degree $\alpha$" and "feature $x$ describes category $\varphi$ to degree $\alpha$", respectively. Here, the interpretation of the propositional connectives $\vee$ and $\wedge$ in the framework described above reproduces the standard notion of join and the meet of fuzzy formal concepts used in fuzzy FCA. The interpretations of the modal operators $\Box$ and $\Diamond$ are motivated by algebraic properties and  duality theory for modal operators on lattices.  

The definitions of $\alpha$-membership and $\alpha$-description given above stand in the same relationship to their crisp counterparts in  \cite{conradie2016categories, conradie2017toward} as Fitting's definition of the  $\alpha$-satisfaction  of modal formulas on many-valued Kripke frame stands to the definition of satisfaction of modal formulas on (crisp) Kripke frames. More precisely, Fitting derives many-valued semantics for modal logic by reading the standard first-order satisfaction clauses for modal formulas as formulas of many-valued predicate logic interpreted on many-valued first-order structures in the standard way (see e.g.\ \cite[Chapter 5]{Hajek1998-HAJMOF}). We proceed similarly. For example, consider the defining clause for the satisfaction of a box-formula on a (crisp) enriched formal context  rewritten in first-order syntax as  
\begin{equation}\label{eqn:Box:satisfaction:crisp}
  \mathfrak{M},  a \Vdash \Box\phi \quad \text{ iff } \quad \forall x ( x \in \descr{\phi}\to a R_{\Box} x).  
\end{equation}
Reading the clause above as a statement of (two-sorted) $\mathbf{H}$-valued predicate logic, the universal quantifier is interpreted as a conjunction indexed by the set $X$ (itself interpreted as a meet in $\mathbf{H}$), the membership statement $x \in \descr{\phi}$ is now $\mathbf{H}$-valued and more naturally written as $\descr{\phi}(x)$, and the same holds for the atomic formula $a R_{\Box} x$, which we write as $R_{\Box}(a,x)$. The implication $\to$ is interpreted as the implication $\to^{\mathbf{H}}$ of $\mathbf{H}$. Furthermore, $\Vdash$ is now interpreted as an  $\mathbf{H}$-valued relation between $H$ and $\mathcal{L}$ and therefore rather than asking if it holds between an object $a$ and a formula, we ask whether it gives a value of at least $\alpha \in \mathbf{H}$ when applied to an object and formula. Thus, \eqref{eqn:Box:satisfaction:crisp} is transformed into 

\begin{equation}\label{eqn:Box:satisfaction:MV}
  \mathfrak{M},  a \Vdash^{\alpha} \Box \phi \quad \text{ iff  } \quad \alpha \leq \bigwedge_{x \in X} ( \descr{\phi}(a) \to^{\mathbf{A}} R_{\Box}(a, x)). 
\end{equation}
Different interpretations of modal operators 
in $\mathrm{LE}$-logic  transfers naturally to the many-valued interpretation. 
 Instead of attributing features to objects absolutely, in this setting, agents can make such attributions in a graded way, and accordingly, their perceived categories consist of stable pairs of $\mathbf{H}$-valued sets of object and features.

\subsection{Non-distributive description logic $\mathrm{LE}$-$\mathcal{ALC}$}
\label{subsec: LE-ALC}


The language of $\mathrm{LE}$-$\mathcal{ALC}$ is intended to be interpreted on the complex algebras of enriched formal contexts. The set of the {\em individual names} for $\mathrm{LE}$-$\mathcal{ALC}$ is the union of two disjoint sets $\mathsf{OBJ}$ and $\mathsf{FEAT}$, which are interpreted as the {\em objects} and {\em features} of the enriched formal contexts, respectively. 
The set $\mathcal{R}$ of the {\em role names} for $\mathrm{LE}$-$\mathcal{ALC}$ is the  union of three types  of  relations: (1) A unique relation  $I \subseteq \mathsf{OBJ} \times \mathsf{FEAT}$; (2) a set of relations $\mathcal{R}_\Box$  of the form  
$ R_\Box \subseteq \mathsf{OBJ} \times \mathsf{FEAT}$; (3) a  set of relations  $\mathcal{R}_\Diamond$  of the form  
$ R_\Diamond \subseteq \mathsf{FEAT} \times \mathsf{OBJ}$. While $I$ is intended to be interpreted as the incidence relation of enriched formal contexts and encodes information on which objects have which features, the relations in $\mathcal{R}_\Box$  and $\mathcal{R}_\Diamond$
encode additional relationships between objects and features (cf.~\cite{conradie2021rough} for an extended discussion). For any set  $\mathcal{C}$ of {\em primitive concepts}, the $\mathrm{LE}$-$\mathcal{ALC}$ {\em concepts} are defined as follows:

{{\centering
 $C \ceq D\ |\ C_1 \wedge C_2\ |\ C_1\vee C_2\ |\ \langle R_\Diamond \rangle C\ |\ [R_\Box ]C$   
\par}}

\noindent where $D \in \mathcal{C}$, $R_\Box \in \mathcal{R}_\Box$ and $R_\Diamond \in \mathcal{R}_\Diamond$. 
Concepts such as $C_1 \vee C_2$ (resp.~$C_1\wedge C_2$) are interpreted as the smallest common superconcept (resp.~the greatest common subconcept) as in FCA. Because there is no canonical and natural way to interpret negations in non-distributive (lattice-based) settings, we do not include $\neg C$ as a concept. Concepts  $\langle R_\Diamond \rangle C$ and $[R_\Box]C$ are interpreted by using the corresponding normal operations (i.e.~$\Diamond$ and $\Box$, respectively) on the complex algebras of enriched formal contexts. We do not include the concept names $\top$ and $\bot$ in the language, as the naive tableaux rules corresponding to these concepts can lead to non-terminating tableaux algorithms. In the future, we intend to extend the results in this paper to the description logic with concept names  $\top$ and $\bot$  in the language. 
We do not use the symbols $\forall r$ and $\exists r$ in the context of $\mathrm{LE}$-$\mathcal{ALC}$, because the semantic clauses of the modal operators in LE-logic use universal quantifiers, and hence using the same notation verbatim would be ambiguous or misleading. 

The {\em TBox axioms} in $\mathrm{LE}$-$\mathcal{ALC}$ are of the shape $C_1 \sqsubseteq C_2$, where $C_1$ and $C_2$ are concepts\footnote{As is standard in DL (cf.~\cite{DLhandbook} for more details), general concept inclusion of the form $C_1 \sqsubseteq C_2$ can be rewritten as concept definition $C_1 \equiv C_2 \wedge C_3$, where $C_3$ is a new concept name.}.
We use $C_1 \equiv C_2$ as a  shorthand for  $C_1 \sqsubseteq C_2$ and $C_2 \sqsubseteq C_1$. 
The {\em ABox terms} in $\mathrm{LE}$-$\mathcal{ALC}$ are of the form:
\noindent

{{\centering
  $aR_\Box x,\quad xR_\Diamond a,\quad aIx,\quad a:C,\quad x::C$  
\par}}
\noindent
The {\em ABox assertions} in $\mathrm{LE}$-$\mathcal{ALC}$ are of the form $t$, $\neg t$, where $t$ is any ABox term. We refer to the terms of first three types and their negations as {\em relational terms}. 
We denote an arbitrary $\mathrm{LE}$-$\mathcal{ALC}$ {\em ABox} (resp.~{\em TBox}) with $\mathcal{A}$ (resp~$\mathcal{T}$). The interpretations of the ABox assertions $a:C$ and  $x::C$ are "object $a$ is a member of concept $C$" and  "feature $x$ is in the description of  concept $C$", respectively. Note that we add the negative terms to ABoxes explicitly, as the $\mathrm{LE}$-$\mathcal{ALC}$ concepts do not contain negation.
An {\em interpretation} for the language of $\mathrm{LE}$-$\mathcal{ALC}$ is a tuple $\mathcal{M} = (\mathfrak{F}, \cdot^\mathcal{M}) $, where $\mathfrak{F}= (A, X, I^\mathcal{M}, \{R_\Box^\mathcal{M} \mid  R_\Box \in \mathcal{R}_\Box \}, \{R_\Diamond^\mathcal{M} \mid  R_\Diamond \in \mathcal{R}_\Diamond \}
)$ is
an enriched formal context and  $\cdot^\mathcal{M}$ maps:

\noindent 1.~each individual name $a \in \mathsf{OBJ}$ (resp.~$x \in \mathsf{FEAT}$),   to some $a^\mathcal{M} \in A$ (resp.~$x^\mathcal{M} \in X$) in $\mathfrak{F}$;

\noindent 2.~role names $I$, $R_\Box \in \mathcal{R}_\Box$ and $R_\Diamond \in \mathcal{R}_\Diamond$ to  relations $I^\mathcal{M}$, $R_\Box^\mathcal{M}$ and $R_\Diamond^\mathcal{M}$ in $\mathfrak{F}$, notice that $R_\Box^\mathcal{M}$ and $R_\Diamond^\mathcal{M}$ in $\mathfrak{F}$ are all $I^\mathcal{M}$-compatible;

\noindent 3.~each primitive concept $D$ to some $D^\mathcal{M} \in \mathfrak{F}^+$, and other concepts as follows:
\smallskip

{{\centering
\small{
    \begin{tabular}{l l ll}
      
    $(C_1 \wedge C_2)^{\mathcal{M}} = C_1^{\mathcal{M}} \wedge C_2^{\mathcal{M}}$&  $(C_1\vee C_2)^{\mathcal{M}} = C_1^{\mathcal{M}} \vee C_2^{\mathcal{M}}$   & $([R_\Box]C)^\mathcal{M} = [R_\Box^\mathcal{M}]C^{\mathcal{M}}$  &  $(\langle R_\Diamond \rangle C)^\mathcal{M} =\langle  R_\Diamond^{\mathcal M} \rangle C^{\mathcal{M}} $\\
    \end{tabular}
    }
\par}}
\smallskip
\noindent
where all the operators are defined as on the complex algebra $\mathfrak{F}^+$. The {\em satisfiability relation} for an interpretation $\mathcal{M}=(\mathfrak{F}, \cdot^\mathcal{M})$ is defined as follows:

\noindent 1.~$\mathcal{M} \vDash C_1\equiv C_2$ iff $\val{C_1^\mathcal{M}} = \val{C_2^\mathcal{M}}$ iff $\descr{C_2^\mathcal{M}} = \descr{C_1^\mathcal{M}}$.

\noindent 2.~$\mathcal{M} \vDash a:C$ iff $a^\mathcal{M} \in \val{C^\mathcal{M}}$, and  $\mathcal{M} \vDash x::C$ iff $x^\mathcal{M} \in \descr{C^\mathcal{M}}$.

\noindent 3.~$\mathcal{M} \vDash  a I x$ (resp.~$a R_\Box x$, $x R_\Diamond a$) iff $a^{\mathcal{M}} I^{\mathcal{M}} x^{\mathcal{M}} $ (resp.~$a^{\mathcal{M}} R_\Box^{\mathcal{M}} x^{\mathcal{M}} $, $x^{\mathcal{M}} R_\Diamond^{\mathcal{M}} a^{\mathcal{M}} $). 

\noindent 4.~$\mathcal{M} \vDash  \neg \alpha$,  where $\alpha$ is any ABox term, iff $ \mathcal{M}\not\vDash \alpha$. 


\noindent An $\mathrm{LE}$-$\mathcal{ALC}$ ABox (resp.~TBox) is a finite set of ABox assertions (resp.~TBox axioms) in $\mathrm{LE}$-$\mathcal{ALC}$. A {\em knowledge base} in $\mathrm{LE}$-$\mathcal{ALC}$ is a tuple $(\mathcal{A}, \mathcal{T})$, where $\mathcal{A}$ is an $\mathrm{LE}$-$\mathcal{ALC}$ ABox, and $\mathcal{T}$ is an $\mathrm{LE}$-$\mathcal{ALC}$ TBox. An interpretation $\mathcal{M}$ is a {\em model} for a knowledge base $(\mathcal{A}, \mathcal{T})$ if $\mathcal{M}\vDash \mathcal{A}$ and $\mathcal{M} \vDash \mathcal{T}$. A knowledege base $(\mathcal{A}, \mathcal{T})$ is {\em consistent} if there is a model for it. An ABox $\mathcal{A}$ (resp.~TBox $\mathcal{T}$) is consistent if there exists a model for knowledege base $(\mathcal{A},\emptyset)$ (resp.~$(\emptyset,\mathcal{T}$)).

%
%
%


\subsection{Tableaux algorithm for checking $\mathrm{LE}$-$\mathcal{ALC}$ ABox consistency}\label{app:LE-ALC}

In this section, we recall the tableaux algorithm for checking $\mathrm{LE}$-$\mathcal{ALC}$ ABox consistency, which is introduced in \cite{van2023old,van2023non}. We noticed a  mistake in the proof of termination and $I$-compatibility in an earlier version of this paper \cite{van2023old} in which concepts $\top$ and $\bot$ were included as concept names in the language. In the updated version \cite{van2023non},  we prove that the result is valid in the restricted language that does not contain $\top$ and $\bot$ in the language. In this paper, we work with the restricted language as in \cite{van2023non}.

An $\mathrm{LE}$-$\mathcal{ALC}$ ABox $\mathcal{A}$ contains a {\em clash} if it contains both ABox assertions $\beta$ and $\neg \beta$. The expansion rules are designed so that the expansion of $\mathcal{A}$ will contain a clash iff  $\mathcal{A}$ is inconsistent. The set $sub(C)$ of the sub-formulas of any $\mathrm{LE}$-$\mathcal{ALC}$ concept $C$ is defined as usual. A concept $C'$ {\em occurs} in $\mathcal{A}$, write as $C' \in \mathcal{A}$, if $C'\in sub(C)$ for some $C$ such that one of the ABox assertions $a:C$,  $x::C$, $\neg a:C$, or $\neg x ::C$ is in $\mathcal{A}$. An individial name $b$ (resp.~$y$) {\em occurs} in $\mathcal{A}$, write as $b \in \mathcal{A}$ (resp.~$y \in \mathcal{A}$), if there is an ABox assertion in $\mathcal{A}$ which contains $b$ (resp.~$y$). 


The Algorithm \ref{alg:LE-ALC}  (constructively) provides a model $\mathcal{M}=(\mathfrak{F},\cdot^\mathcal{M})$ for every consistent ABox $\mathcal{A}$, where $\mathfrak{F}= (A, X, I, \mathcal{R}_\Box, \mathcal{R}_\Diamond)$. This model has the following property:
\begin{center}
For any $C \in \mathcal{A}$, there exist $a_C \in A$ and $x_C \in X$ such that, for any $a \in A$ (resp.~$x \in X$), $a \in \val{C^{\mathcal{M}}}$ (resp.~$x \in \descr {C}^{\mathcal{M}}$) iff $a I x_C$ (resp.~$a_C I x$). 
\end{center}

\noindent
We call $a_C$ and $x_C$ the  {\em classifying object} and the {\em classifying feature} of $C$, respectively. To make our notation more  readable, we will write $a_{\Box C}$, $x_{\Box C}$ (resp.~$a_{\Diamond C}$, $x_{\Diamond C}$) instead of $a_{[R_\Box]C}$, $x_{[R_\Box]C}$ (resp.~$a_{\langle R_\Diamond\rangle C}$, $x_{\langle R_\Diamond\rangle C}$).

\begin{algorithm} 
\caption{tableaux algorithm for checking $\mathrm{LE}$-$\mathcal{ALC}$ ABox consistency }\label{alg:LE-ALC}
\label{alg:tableaux}
    \hspace*{\algorithmicindent} \textbf{Input}: An   $\mathrm{LE}$-$\mathcal{ALC}$ ABox $\mathcal{A}$. \quad \textbf{Output}: whether $\mathcal{A}$ is inconsistent. 
    \begin{algorithmic}[1]
        \State \textbf{if} there is a clash in $\mathcal{A}$ \textbf{then} \textbf{return} ``inconsistent''.
        \State \textbf{pick} any applicable expansion rule $R$, \textbf{apply} $R$ to $\mathcal{A}$ and proceed recursively.
      \State \textbf{if}  no expansion  rule is applicable   \textbf{return} "consistent".
    \end{algorithmic}
\end{algorithm}

{{\footnotesize
\centering
\begin{tabular}{cc} 
\mc{1}{c}{\textbf{Creation rule}} & \mc{1}{c}{\textbf{Basic rule}} \\
\AXC{For any $C \in \mathcal{A}$}
\RL{\fns create}
\UIC{$a_C:C$, \quad $x_C::C$}
\DP
 \ & \ 
\rule[-1.85mm]{0mm}{8mm}
\AXC{$b:C, \quad y::C$}
\LL{\fns $I$}
\UIC{$b I y$}
\DP 


\end{tabular}

\begin{tabular}{cccc}
\mc{2}{c}{\textbf{Rules for the logical connectives}} & \mc{2}{c}\textbf{$I$-compatibility rules} \\
\rule[-1.85mm]{0mm}{8mm}
 \AXC{$b:C_1 \wedge  C_2$}
\RL{\fns $\wedge_A$}
\UIC{$b:C_1,$ \quad $b:C_2$}
\DP 
\ & \

\AXC{$y::C_1 \vee  C_2$}
\LL{\fns $\vee_X$}
\UIC{$y::C_1,$ \quad $y::C_2$}
\DP 
\  &  \

\AXC{$b I \Box y$}
\LL{\fns $\Box y$}
\UIC {$b R_\Box y$}
\DP

\ & \

\AXC{$b I \blacksquare y$}
\RL{\fns $\blacksquare y$}
\UIC {$y R_\Diamond b$}
\DP
\\[3mm]

 \AXC{$b:[R_\Box]C,$ \quad  $y::C$}
\LL{\fns $\Box$}
\UIC{$b R_\Box y$}
\DP 
\ & \ 
  \AXC{$y::\langle  R_\Diamond \rangle C,$ \quad  $b:C$}
\RL{\fns $\Diamond$}
\UIC{$y R_\Diamond b$}
\DP 
\ & \
\AXC{$\Diamond b I  y$}
\LL{\fns $\Diamond b$}
\UIC {$y R_\Diamond b$}
\DP
\ & \
\AXC{$\Diamondblack b I y$}
\RL{\fns $\Diamondblack b$}
\UIC {$b R_\Box y$}
\DP
\\[3 mm]

\end{tabular}
\begin{tabular}{rl}
\mc{2}{c}{\textbf{invese rule for connectives}}  \\

   \AXC{$b:C_1$, $b:C_2$, $C_1 \wedge C_2 \in \mathcal{A}$}
\LL{\fns $\wedge_A^{-1}$}
\UIC {$b:C_1 \wedge C_2$}
\DP
\ & \
\AXC{$y::C_1$, $y::C_2$, $C_1 \vee C_2 \in \mathcal{A}$}
\RL{\fns $\vee_X^{-1}$}
\UIC {$y::C_1 \vee C_2$}
\DP
\end{tabular}

\begin{tabular}{rl}
\mc{2}{c}{\textbf{Adjunction rules}}  \\
\rule[-1.85mm]{0mm}{8mm}
\AXC{$ b R_\Box y$}
\LL{\fns $R_\Box$}
\UIC{$\Diamondblack b I y,$ \quad  $b I \Box y$}
\DP 
\ & \ 
\AXC{$y R_\Diamond b$}
\RL{\fns $R_\Diamond$}
\UIC{$\Diamond b I y,$ \quad  $b I \blacksquare y$}
\DP 
\\ [2mm]
\end{tabular}

\begin{tabular}{cccc}
\mc{2}{c}{\textbf{Basic rules for negative assertions}} & \mc{2}{c}{\textbf{Appending  rules}} \\
\rule[-1.85mm]{0mm}{8mm}
\AXC{$\neg (b:C)$}
\LL{\fns $\neg b$}
\UIC{$\neg (b I x_C)$}
\DP 
\ & \ 
\AXC{$\neg (x::C)$}
\RL{\fns $\neg x$}
\UIC{$\neg (a_C I x)$}
\DP 
\ & \ 
\AXC{$b I x_C$}
\LL{\fns $x_C$}
\UIC{$b:C$}
\DP 
\ & \ 
\AXC{$a_C I y$}
\RL{\fns $a_C$}
\UIC{$y::C$}
\DP 
\\
\end{tabular}
\par}}
\smallskip 

\noindent In the adjunction rules the individuals $\Diamondblack b$, $\Diamond b$, $\Box y$, and $\blacksquare y$ are new and unique for each relation $R_\Box$ and $R_\Diamond$, except for $\Diamond a_C= a_{\Diamond C}$ and 
$\Box x_C= x_{\Box C}$\footnote{The new individual names $\Diamondblack b$, $\Diamond b$, $\Box y$,   and $\blacksquare y$ appearing in tableaux expansion are purely syntactic entities.  Intuitively, they correspond to the classifying objects (resp.~features) of the concepts $\Diamondblack \textbf{b}$, $\Diamond \textbf{b}$ (resp.~$\Box \textbf{y}$, resp.~$\blacksquare \textbf{y}$),  where $\mathbf{b}=(b^{\uparrow\downarrow}, b^\uparrow)$ (resp.~$\mathbf{y}=(y^\downarrow, y^{\downarrow\uparrow})$) is the concept generated by $b$ (resp.~$y$),  and  the operation  $\Diamondblack$ (resp.$\blacksquare$) is the left (resp.~right) adjoint of operation $\Box$ (resp.~$\Diamond$).}. 

The following theorem follows from the results in \cite{van2023non}:
\begin{theorem}\label{thm:ABox LE-ALC}
Algorithm \ref{alg:LE-ALC} provides a sound and complete polynomial-time decision procedure for checking consistency of  $\mathrm{LE}$-$\mathcal{ALC}$ ABoxes. 
\end{theorem}

\begin{remark}\label{acyclic TBoxes}
    The Algorithm \ref{alg:LE-ALC} can be extended to  an exponential-time algorithm for checking consistency of knowledge bases with acyclic TBoxes via the unraveling technique (cf.~\cite{DLhandbook} for details). 
\end{remark}

 \section{Description logic $\mathrm{LE}$-$\mathcal{FALC}$}\label{sec:LE-FALC}
 
In this section, we introduce the fuzzy non-distributive description logic $\mathrm{LE}$-$\mathcal{FALC}$ based on many-valued  $\mathrm{LE}$-logic, which can be seen as a fuzzy generalization of description logic $\mathrm{LE}$-$\mathcal{ALC}$. 

The {\em individual names}, {\em role names}, {\em primitive concepts}, {\em concepts} and {\em ABox terms} in the language of $\mathrm{LE}$-$\mathcal{FALC}$ are the same as those in the language of $\mathrm{LE}$-$\mathcal{ALC}$ (recall Section \ref{subsec: LE-ALC}), except that they are intended to be interpreted on the complex algebras of $\mathbf{H}$-valued enriched formal contexts as described in Section \ref{subsec:Many-valued polarity-based semantics}. 
The {\em ABox assertions} in $\mathrm{LE}$-$\mathcal{FALC}$ are of the form:

{{\centering
$\alpha \leq t, \quad  \alpha \not \leq t$,  
\par}}
\noindent where $\alpha \in \mathbf{H}$, and $t$ is an ABox term.  
For any of the ABox terms $t$, $\alpha \leq t$ (resp.~$\alpha \not\leq t$) is interpreted as $t$ is true  at least to the extent $\alpha$ (resp.~it is not the case that $t$ is true  at least to the extent $\alpha$).
The {\em TBox axioms in $\mathrm{LE}$-$\mathcal{FALC}$} are of the form:

{{\centering
$C_1\equiv C_2,$
\par}}
\noindent where $C_1$ and $C_2$ are $\mathrm{LE}$-$\mathcal{FALC}$ concepts. 


  An {\em interpretation} for $\mathrm{LE}$-$\mathcal{FALC}$ knowledge base is a tuple $\mathcal{M} = (\mathfrak{F}, \cdot^\mathcal{M})$, where $\mathfrak{F} =(A, X, I^\mathcal{M}, \{R_\Box^\mathcal{M} \mid  R_\Box \in \mathcal{R}_\Box \}, \{R_\Diamond^\mathcal{M} \mid  R_\Diamond \in \mathcal{R}_\Diamond \}
)$ is an $\mathbf{H}$-valued enriched formal context, and   $\cdot^\mathcal{M}$ maps:

\noindent 1.~any individual name $a \in \mathsf{OBJ}$ (resp.~$x \in \mathsf{FEAT}$) to some $a^\mathcal{M} \in A$ (resp.~$x^\mathcal{M} \in X$);

\noindent 2.~role names $I$, $R_\Box \in \mathcal{R}_\Box$ and $R_\Diamond \in \mathcal{R}_\Diamond$ to $\mathbf{H}$-relations 
$I^\mathcal{M}$, $R_\Box^\mathcal{M}$ and $R_\Diamond^\mathcal{M}$ in $\mathfrak{F}$, notice that $R_\Box^\mathcal{M}$ and $R_\Diamond^\mathcal{M}$ in $\mathfrak{F}$ are all $I^\mathcal{M}$-compatible;

\noindent 3.~any primitive concept $D$ to  $D^\mathcal{M} \in \mathfrak{F}^+$, and the other concepts as follows:
\smallskip

{{\centering
\small{
    \begin{tabular}{l l ll}
      
    $(C_1 \wedge C_2)^{\mathcal{M}} = C_1^{\mathcal{M}} \wedge C_2^{\mathcal{M}}$&  $(C_1\vee C_2)^{\mathcal{M}} = C_1^{\mathcal{M}} \vee C_2^{\mathcal{M}}$   & $([R_\Box]C)^\mathcal{M} = [R_\Box^\mathcal{M}]C^{\mathcal{M}}$  &  $(\langle R_\Diamond \rangle C)^\mathcal{M} =\langle  R_\Diamond^{\mathcal M} \rangle C^{\mathcal{M}} $,\\
    \end{tabular}
    }
\par}}
\smallskip
\noindent where all the connectives are interpreted like in the many-valued $\mathrm{LE}$-logic. For any interpretation $\mathcal{M}=(\mathfrak{F}, \cdot^\mathcal{M})$, we define  an $\mathbf{H}$-valuation $v^\mathcal{M}$ for  ABox terms as follows:

    
    \noindent 1.~$v^\mathcal{M}(a:C) = \val{C}(a)$ and $v^\mathcal{M}(x::C) = \descr{C}(x)$. 
    
    \noindent 2.~$v^\mathcal{M}(I(a,x))= I^\mathcal{M}(a^\mathcal{M},x^\mathcal{M})$, $v^\mathcal{M}( R_\Box(a,x))= R_\Box^\mathcal{M}(a^\mathcal{M},x^\mathcal{M})$, and $v^\mathcal{M}(R_\Diamond(x,a))= R_\Diamond^\mathcal{M}(x^\mathcal{M},a^\mathcal{M})$.

 Let $\mathcal{M}$ be an interpretation, the {\em satisfiability relation} for $\mathcal{M}$ is defined as follows:

\noindent 1.~$\mathcal{M} \vDash C_1\equiv C_2$ iff $ C^\mathcal{M}_1= C^\mathcal{M}_2$. 

\noindent 2.~For any ABox term $t$,  $\mathcal{M} \vDash \alpha \leq t$ iff $\alpha \leq v^\mathcal{M}(t)$ and  $\mathcal{M} \vDash  \alpha \not \leq t $ iff $\mathcal{M} \not\vDash \alpha \leq t$.

\noindent A {\em knowledge base} in $\mathrm{LE}$-$\mathcal{FALC}$ is a tuple $(\mathcal{A}, \mathcal{T})$, where $\mathcal{A}$ is an $\mathrm{LE}$-$\mathcal{ALC}$ ABox, and $\mathcal{T}$ is an $\mathrm{LE}$-$\mathcal{ALC}$ TBox. An interpretation $\mathcal{M}$ is a {\em model} for a knowledge base $(\mathcal{A}, \mathcal{T})$ if $\mathcal{M}\vDash \mathcal{A}$ and $\mathcal{M} \vDash \mathcal{T}$. A knowledge base $(\mathcal{A}, \mathcal{T})$ is {\em consistent} if there is a model for it. An ABox $\mathcal{A}$ (resp.~TBox $\mathcal{T}$) is consistent if there exists a model for knowledge base $(\mathcal{A},\emptyset)$ (resp.~$(\emptyset,\mathcal{T}$)).

\section{Examples of $\mathrm{LE}$-$\mathcal{FALC}$ knowledge bases}\label{sec:example of knowledge bases}

We now give two toy examples of  knowledge bases regarding categorization of scientific papers into different  categories based on keywords appearing in them  along with epistemic understanding  of resulting categories according to  an agent in the language  $\mathrm{LE}$-$\mathcal{FALC}$.  

\noindent The following table lists concepts appearing in the knowledge base.

{{
\par
\centering
 \begin{tabular}{|c|c|c|c|}
    \hline
        \textbf{concept name} & \textbf{symbol} & 
        \textbf{concept name} & \textbf{symbol}   \\
        \hline
          Chemistry papers & $C_1$ & Physics papers & $C_2$ \\
           Biology papers & $C_3$ & Natural science papers& $C_4$ \\
          Biochemistry papers  & $C_5$  & Multi-disciplinary physics papers & $C_6$ \\
        \hline
    \end{tabular}
    \par 
}}
\smallskip

\noindent This knowledge base is based on the following set of features:

{{
\par
\centering
 \begin{tabular}{|c|c|c|c|}
    \hline
        \textbf{feature} & \textbf{symbol} & \textbf{feature} & \textbf{symbol}   \\
        \hline
        keywords from chemistry & $y_1$ &   keywords from physics & $y_2$ \\
         keywords from biology& $y_3$  & keywords from biochemistry & $y_4$ \\
         
           \hline 
           \end{tabular}
           \par
}}
\smallskip
Let $\mathbf{H}$ be the Heyting algebra defined by a totally ordered set $\{0, 1/2, 1\}$. As $\mathbf{H}$ is linear, we use $t<\alpha$ to denote the ABox assertion $\alpha\not\leq t$ for any ABox term $t$ and $\alpha\in\mathbf{H}$. For any paper $P$ and a feature $y$, $I(P,y)= 0$ (resp.~$I(P,y)= 1/2$, resp.~$I(P,y)= 1$) if it contains very few or no (resp.~some, resp.~many) keywords  from the field corresponding to feature $y$.  For example, $I(P,y_1)= 1/2$ if  paper $P$ contains some (but not many) keywords from chemistry.
Let $R_\Box$  be a relation corresponding to an epistemic agent $i$ defined as follows: For any paper $P$, and a feature $y$, $R_\Box(P,y)= 0$ (resp.~$R_\Box(P,y)= 1/2$, resp.~$R_\Box(P,y)= 1$) if it contains very few  or no (resp.~some, resp.~many) keywords from the field corresponding to feature $y$ according to agent $i$.  The $I$-compatibility of relation $R_\Box$ is justified as explained in Section \ref{subsec:Many-valued polarity-based semantics}. For any two categories  $C_1$ and $C_2$, the categories $C_1 \vee C_2$ and $C_1 \wedge C_2$ denote their least common super-category, and greatest common sub-category, respectively. The category $[R_\Box] C_1$, denotes the formal concept generated by the (fuzzy) set of objects that agent $i$ believes to have all features to the extent specified by intension of  $C_1$ i.e.~$[R_\Box] C_1$ can be seen as a concept whose extension is perception of extension of $C_1$ (given the intension of $C_1$) according to agent $i$.

Let $\mathcal{K} = (\mathcal{A}, \mathcal{T})$ be a knowledge base such that  $\mathcal{T} =
\{ C_4  \equiv C_1 \vee  C_3,   C_2 \equiv   C_4 \wedge C_7, C_6 \equiv (C_1 \wedge C_2) \vee (C_2 \wedge C_3) 
\}$,  and $\mathcal{A}=\{1 \leq P_1:C_2, P_1:C_6< 1/2, 1 \leq y_1::[R_\Box] C_1,  1 \leq y_3::[R_\Box] C_3, P_2 R_\Box y_3<1,  1/2\leq P_2:C_1, 1/2 \leq y_1::C_1  \}$, where $C_7$ is a  new  fresh concept name. TBox axiom $C_4\equiv C_1 \vee  C_3$ means that the category of natural science papers is the smallest category (in the categorization system defined by given knowledge base) of papers containing both chemistry and biology papers. TBox axiom $C_2\equiv C_4 \wedge C_7$ states that the physics papers are natural science papers. 
 TBox axiom  $C_6  \equiv (C_1 \wedge C_2) \vee (C_2 \wedge C_3) $ states that multi-disciplinary physics papers is the smallest category containing categories of papers which are considered both physics and chemistry papers and physics and biology papers.
 
 ABox assertion $1 \leq P_1:C_2 $ (resp.~$1/2 \leq P_2:C_1 $) states that  paper $P_1$ (resp.~$P_2$) is a physics (resp.~chemistry) paper to the extent $1$ (resp.~$1/2$). ABox assertion $P_1:C_1 < 1/2$ (resp.~$P_1:C_6 < 1/2$) state that  paper $P_1$ is a chemistry (resp.~multi-disciplinary physics) paper to extent $0$. ABox assertion $1 \leq y_1::[R_\Box] C_1$, 
 (resp.~$1 \leq y_3::[R_\Box] C_3$) states that any paper which is Chemistry (resp.~biology) paper according to agent $i$  must have many keywords from chemistry (resp.~biology). The ABox assertion $ P_2 R_\Box y_3<1$ states that paper $P_2$ does not have many physics keywords according to agent $i$. ABox assertion $1/2 \leq  y_1::C_1 $ states that having at least some keywords from Chemistry is in intension (defining features) of category of Chemistry papers. Note that the assertions $ C_2  \equiv   C_4 \wedge C_7$, $C_6  \equiv (C_1 \wedge C_2) \vee (C_2 \wedge C_3) $, $1 \leq P_1:C_2 $, and 
  $P_1:C_6 < 1/2$ together imply that any model of $\mathcal{K}$ must be non-distributive.

For another example, consider the knowledge base $\mathcal{K}'=(\mathcal{A}', \mathcal{T}')$ such that $\mathcal{T}'=\{C_5 \equiv C_1 \wedge C_3\}$ and $\mathcal{A}'=\{1/2 \leq P_3:[R_\Box] C_1, 1 \leq P_3:[R_\Box] C_3,  P_3 R_\Box y_4 <1/2,  1/2 \leq y_4:: C_5\}$. We can provide interpretations for different axioms of $\mathcal{K}'$ similar to the interpretations of axioms in  $\mathcal{K}$  discussed above.

\section{Tableaux algorithm for checking consistency of ABoxes}\label{Sec: tableau}

In this section, we introduce the Algorithm \ref{alg:main algo} which is modified from the Algorithm \ref{alg:LE-ALC} to check the consistency of $\mathrm{LE}$-$\mathcal{FALC}$ ABoxes. We say that $\mathrm{LE}$-$\mathcal{FALC}$ ABox $\mathcal{A}$ contains a {\em clash} if there exist  $\alpha_1 \leq t$ and $ \alpha_2 \not \leq t$ in  $\mathcal{A}$ with $\alpha_2 \leq \alpha_1$ for some ABox term $t$.
The set $sub(C)$ of the sub-formulas of any $\mathrm{LE}$-$\mathcal{FALC}$ concept $C$ is defined as usual. A concept $C'$ {\em occurs} in $\mathcal{A}$, write as $C' \in \mathcal{A}$, if $C'\in sub(C)$ for some $C$ such that  there is an ABox assertion in $\mathcal{A}$ which contains $C$. 
 A constant $b$  (resp.~$y$) {\em occurs} in $\mathcal{A}$, in symbols,  $b \in \mathcal{A}$ (resp.~$y \in \mathcal{A}$), if there is an ABox assertion in $\mathcal{A}$ which contains $b$ (resp.~$y$).

For any concept $C \in \mathcal{A}$, the added constants $a_C$ and $x_C$ act as {\em classifying object} and {\em classifying feature} of concept $C$, in the sense that if $\mathcal{A}$ is consistent, then the tableaux algorithm will construct a model which satisfies $\val{C}(b) = \max \{ \alpha \mid \alpha \leq  I (b,x_C) \in  \overline{\mathcal{A}}\}$, and $ \descr{C}(y) = \max \{ \alpha \mid \alpha \leq  I (a_C ,y) \in  \overline{\mathcal{A}}\}$ for any object $b$, and feature $y$ in the model.  The set of tableaux expansion rules for   $\mathrm{LE}$-$\mathcal{FALC}$ is listed below. The commas in each rule are  meta-linguistic conjunctions, hence every tableau is non-branching. 

\smallskip

{{\footnotesize
\centering
\begin{tabular}{cc} 
\mc{1}{c}{\textbf{Creation rule}} & \mc{1}{c}{\textbf{Basic rule}} \\
\AXC{For any $C \in \mathcal{A}$}
\RL{\fns create}
\UIC{$1 \leq a_C:C$, \quad $1 \leq x_C::C$}
\DP
 \ & \ 
\rule[-1.85mm]{0mm}{8mm}
\AXC{$\alpha_1 \leq b:C, \quad \alpha_2 \leq y::C$}
\LL{\fns $I$}
\UIC{$\alpha_1 \wedge \alpha_2 \leq I(b,y)$}
\DP 


\end{tabular}

\begin{tabular}{cccc}
\mc{2}{c}{\textbf{Rules for the logical connectives}} & \\
\rule[-1.85mm]{0mm}{8mm}
 \AXC{$\alpha \leq b:C_1 \wedge  C_2$}
\RL{\fns $\wedge_A$}
\UIC{$\alpha \leq  b:C_1,$ \quad $\alpha \leq  b:C_2$}
\DP 
\ & \

\AXC{$\alpha \leq y::C_1 \vee  C_2$}
\LL{\fns $\vee_X$}
\UIC{$\alpha \leq y::C_1,$ \quad $\alpha \leq  y::C_2$}
\DP 
\  &  \

\\[3mm]

 \AXC{$\alpha_1 \leq  b:[R_\Box]C,$ \quad  $\alpha_2 \leq  y::C$}
\LL{\fns $\Box$}
\UIC{$\alpha_1 \wedge \alpha_2  \leq R_\Box (b,y)$}
\DP 
\ & \ 
  \AXC{$\alpha_1 \leq  y::\langle  R_\Diamond \rangle C,$ \quad  $\alpha_2 \leq  b:C$}
\RL{\fns $\Diamond$}
\UIC{$\alpha_1 \wedge \alpha_2 \leq   R_\Diamond (y,b)$}
\DP 
\\[3 mm]
\end{tabular}

\begin{tabular}{cccc}
 \mc{3}{c}{$I$-\textbf{compatibility rules}} \\
 \AXC{$\alpha \leq  I(b,\Box y)$}
\LL{\fns $\Box y$}
\UIC {$\alpha \leq   R_\Box (b,y) $}
\DP
\ & \
\AXC{$\alpha \leq I( b, \blacksquare y)$}
\RL{\fns $\blacksquare y$}
\UIC {$\alpha \leq   R_\Diamond (y,b)$}
\DP
\ & \
\AXC{$\alpha \leq  I(\Diamond b ,y)$}
\LL{\fns $\Diamond b$}
\UIC {$\alpha \leq R_\Diamond (y,b)$}
\DP
\ & \
\AXC{$\alpha \leq I(\Diamondblack b, y)$}
\RL{\fns $\Diamondblack b$}
\UIC {$c R_\Box (b,y)$}
\DP
\\
\end{tabular}

\begin{tabular}{rl}
\mc{2}{c}{\textbf{inverse rule for connectives}}  \\

   \AXC{$\alpha_1 \leq b:C_1$, $\alpha_2 \leq b:C_2$; where $C_1 \wedge C_2 \in \mathcal{A}$}
\LL{\fns $\wedge_A^{-1}$}
\UIC {$\alpha_1 \wedge \alpha_2 \leq b:C_1 \wedge C_2$}
\DP
\ & \
\AXC{$\alpha_1 \leq y::C_1$, $\alpha_2 \leq y::C_2$; where $C_1 \vee C_2 \in \mathcal{A}$}
\RL{\fns $\vee_X^{-1}$}
\UIC {$\alpha_1 \wedge \alpha_2 \leq  y::C_1 \vee C_2$}
\DP
\end{tabular}

\begin{tabular}{rl}
\mc{2}{c}{\textbf{Adjunction rules}}  \\
\rule[-1.85mm]{0mm}{8mm}
\AXC{$ \alpha \leq R_\Box (b,y)$}
\LL{\fns $R_\Box$}
\UIC{$\alpha \leq I(\Diamondblack b, y)$, \quad  $\alpha \leq I(b,\Box y)$}
\DP 
\ & \ 
\AXC{$\alpha \leq R_\Diamond (y,b)$}
\RL{\fns $R_\Diamond$}
\UIC{$\alpha \leq I(\Diamond b,y),$ \quad  $\alpha \leq I(b, \blacksquare y)$}
\DP 
\\ [2mm]
\end{tabular}

\begin{tabular}{cccc}
\mc{2}{c}{\textbf{Basic rules for  negative assertions}} & \mc{2}{c}{\textbf{Appending  rules}} \\
\rule[-1.85mm]{0mm}{8mm}
\AXC{$  \alpha \not \leq (b:C)  $}
\LL{\fns $-a_C$}
\UIC{$ \alpha \not \leq I(b,x_C)$}
\DP 
\ & \ 
\AXC{$\alpha \not \leq   (y::C)$}
\LL{\fns $-x_C$}
\UIC{$\alpha \not \leq   I(a_C, y) $}
\DP 
\ & \ 
\AXC{$\alpha \leq I(b,x_C)$}
\LL{\fns $x_C$}
\UIC{$\alpha \leq b:C$}
\DP 
\ & \ 
\AXC{$\alpha \leq I(a_C,y)$}
\RL{\fns $a_C$}
\UIC{$\alpha \leq y::C$}
\DP 
\\
\end{tabular}
\par}}
\smallskip

{{\centering
\begin{tabular}{cc}
\mc{2}{c}{\textbf{Many-valued algebra rules}}  \\
\AXC{$\alpha_1 \leq  t$, \quad $\alpha_2 \leq  t$}
\LL{\fns $MV_{\vee}$}
\UIC{$\alpha_1 \vee \alpha_2 \leq t$}
\DP   

\end{tabular}

\par}}
\smallskip 

\noindent In the adjunction rules the individual names $\Diamondblack b$, $\Diamond b$, $\Box y$,  and $\blacksquare y$ must be new and unique for each relation $R_\Box$ and $R_\Diamond$, except for  $\Diamond a_C= a_{\Diamond C}$ and 
$\Box x_C= x_{\Box C}$\footnote{The new individual names $\Diamondblack b$, $\Diamond b$, $\Box y$,   and $\blacksquare y$ appearing in tableaux expansion are purely syntactic entities.  Intuitively, they correspond to the classifying objects (resp.~features) of the concepts $\Diamondblack \textbf{b}$, $\Diamond \textbf{b}$ (resp.~$\Box \textbf{y}$, resp.~$\blacksquare \textbf{y}$),  where $\mathbf{b}=(b^{\uparrow\downarrow}, b^\uparrow)$ (resp.~$\mathbf{y}=(y^\downarrow, y^{\downarrow\uparrow})$) is the concept generated by $b$ (resp.~$y$),  and  the operation  $\Diamondblack$ (resp.$\blacksquare$) is the left (resp.~right) adjoint of operation $\Box$ (resp.~$\Diamond$).}. Side conditions for rules $\wedge_A^{-1}$, and $\vee_X^{-1}$ ensure we do not add new  joins or meets to concept names.

\begin{algorithm} 
\caption{tableaux algorithm for checking $\mathrm{LE}$-$\mathcal{FALC}$ ABox consistency }\label{alg:main algo}
    \hspace*{\algorithmicindent} \textbf{Input}: An $\mathrm{LE}$-$\mathcal{FALC}$ ABox $\mathcal{A}$. \quad \textbf{Output}: whether $\mathcal{A}$ is inconsistent. 
    \begin{algorithmic}[1]
        \State \textbf{if} there is a clash in $\mathcal{A}$ \textbf{then} \textbf{return} ``inconsistent''.
        \State \textbf{pick} any applicable expansion rule $R$, \textbf{apply} $R$ to $\mathcal{A}$ and proceed recursively.
      \State \textbf{if}  no expansion  rule is applicable   \textbf{return} ``consistent''.
    \end{algorithmic}
\end{algorithm}

From the shape of the expansion rules, it is clear that the new terms are added by $\mathrm{LE}$-$\mathcal{FALC}$  expansion rules in a manner identical to  $\mathrm{LE}$-$\mathcal{ALC}$ expansion rules except for the rule $MV_\vee$.  The rule $MV_\vee$ at most adds one term for two terms already present in the tableau. Hence, this rule can only increase the size of tableau linearly.  Moreover, application of any of the tableaux rules only involves application of a fixed (finite) number of $\mathbf{H}$ operations. Hence, the polynomial time termination result for $\mathrm{LE}$-$\mathcal{FALC}$ tableaux algorithm  follows  from the termination result for  $\mathrm{LE}$-$\mathcal{ALC}$ tableaux algorithm (See \cite[Section 4.1] {van2023non}, for more details).

\begin{theorem}[Termination]\label{thm:termination}
For any ABox $\mathcal{A}$, the tableaux algorithm \ref{alg:main algo} terminates in a finite number of steps which is polynomial in $size(\mathcal{A})$.
\end{theorem}

Similarly to $\mathrm{LE}$-$\mathcal{ALC}$ (see Remark \ref{acyclic TBoxes}),  Algorithm \ref{alg:main algo} can be extended to acyclic TBoxes (exponential-time) via  unraveling. In the following sections, we prove the completeness and soundness for the Algorithm \ref{alg:main algo}.

\section{Soundness of the tableaux algorithm}\label{sec:soundness}

 For any consistent ABox $\mathcal{A}$, we let its {\em completion} $\overline{\mathcal{A}}$ be its maximal expansion (which exists due to termination). 
 If there is no clash in $\overline{\mathcal{A}}$, we construct a tuple $\mathcal{M}=(\mathfrak{F},\cdot^\mathcal{M})$ as follows:

 \noindent 1.~$A$ and $X$ are the sets of individual names of objects and features occurring in $\overline{\mathcal{A}}$, respectively. 
 
\noindent 2.~For any $a \in A$, $x \in X$, and any role names $R_\Box \in \mathcal{R}_\Box$, $R_\Diamond \in \mathcal{R}_\Diamond$
we have the values of the maps $I(a,x)$, $R_\Box(a,x)$, $R_\Diamond (a,x)$ are the maximum $\alpha_i$ such that ABox assertions $\alpha_i \leq I(a,x)$, $\alpha_i \leq R_\Box(a,x)$, $\alpha_i \leq R_\Diamond (a,x)$ explicitly occur in  $\overline{\mathcal{A}}$. 

\noindent 3.~Let $\mathfrak{F}= (A,X,I, \mathcal{R}_\Box, \mathcal{R}_\Diamond)$ be the tuple obtained in this manner. 

\noindent 4.~We  add a new element $x_\bot$ (resp.~$a_\top$) to $X$ (resp.~$A$) such that it is related to any element of $A$ (resp.~$X$) to extent $0$ w.r.t.~all the relations.

\noindent 5.~The map $\cdot^\mathcal{M}$ is defined as follows: For any individual name $a$ (resp.~$x$), we let $a^\mathcal{M}\coloneqq a$ (resp.~$x^\mathcal{M}\coloneqq x$). For any primitive concept $D\in\mathcal{A}$, we define $D^{\mathcal{M}}= ({\{1/x_D\}}^\downarrow, {\{1/a_D\}}^\uparrow)$.

Next, we show that $\mathcal{M}$ is a model for the ABox ${\mathcal{A}}$. To this end, we need to show that $\mathfrak{F}$ is an $\mathbf{H}$-valued enriched formal context, i.e.~that  $R_\Box$ and $R_\Diamond$ are $I$-compatible, and  $D^{\mathcal{M}}$ is a fuzzy formal concept in the complex algebra $\mathfrak{F}^+$ for every atomic concept $D$. The latter is shown in the next lemma, and the former in the subsequent one.
\begin{lemma}\label{lem:atomic concepts}
$\{1/x_D\}^{\downarrow\uparrow}=\{1/a_D\}^\uparrow$ and $\{1/a_D\}^{\uparrow\downarrow}=\{1/x_D\}^\downarrow$ for any primitive concept $D\in\mathcal{A}$.
\end{lemma}
\begin{proof}
    We only prove the first equation here, and the second equation can be proved similarly. It sucffices to prove that for any $y\in X$, $\{1/x_D\}^{\downarrow\uparrow}(y)=\{1/a_D\}^\uparrow(y)$. Note that for any $y \in X$, and $b \in A$, 
\begin{equation}\label{eq:closure characteristic A}
\{1/x_D\}^{\downarrow\uparrow}(y)= \wedge_{b \in A} (I(b,x_D)\rightarrow I(b,y)),
\end{equation}
\begin{equation}\label{eq:closure characteristic X}
\{1/a_D\}^{\downarrow\uparrow}(b)= \wedge_{y\in X} (I(a_D,y) \rightarrow I(b,y)).
\end{equation}
By the creation rules, we always have $1 \leq a_D:D$ and $1 \leq x_D::D$ in $\overline{\mathcal{A}}$.  Then by rule $I$, we have  that $1 \leq  I(a_D, x_D)$ in $\overline{\mathcal{A}}$.  Then by construction and properties of $\mathbf{H}$, for any $y \in X$, $I(a_D,x_D)\rightarrow I(a_D,y) = I(a_D,y)$. Therefore, by Equation \eqref{eq:closure characteristic A}, we have $\{1/x_D\}^{\downarrow\uparrow}(y) \leq I(a_D,y) =\{1/a_D\}^\uparrow(y)$.

For the reverse direction, suppose $\alpha_1 \leq \{1/a_D\}(y)=  I(a_D,y)$. We need to show that $\alpha_1 \leq \{1/x_D\}^{\downarrow\uparrow}(y)$ for every $y$ occurring in $\overline{\mathcal{A}}$.  By Equation \eqref{eq:closure characteristic A}, it is enough to show that for every $b \in A$, $ \alpha_1  \leq I (b,x_D) \rightarrow I(b,y)$. By adjunction, this is equivalent to showing $ \alpha_1  \wedge  I (b,x_D) \leq  I(b,y)$
for every $b \in A$.    If term $b I x_D$ does not occur in $\overline{\mathcal{A}}$, by construction, we have $ I(b,x_D) =0$ which trivially implies the required condition. Suppose $ I(b,x_D) =\alpha_2 \neq 0$, then by construction $\alpha_2 \leq I(b,x_D) \in \overline{\mathcal{A}}$.   As $\alpha_1 \leq  I(a_D,y)$, by construction $\alpha_3 \leq  I(a_D,y) \in \overline{\mathcal{A}}$ for some  $\alpha_1 \leq \alpha_3$. 
Therefore, by rule $I$, we have $\alpha_2 \wedge \alpha_3 \leq  I(b,y) \in  \overline{\mathcal{A}}$. Then by construction, we have $\alpha_1 \wedge \alpha_2\leq \alpha_2 \wedge \alpha_3 \leq  I(b,y)$. Hence proved.
\end{proof}

\begin{lemma} \label{lem:Galois-stability} 
The relations $R_\Box \in \mathcal{R}_\Box$ and $R_\Diamond \in \mathcal{R}_\Diamond$ in $\mathfrak{F}= (\mathfrak{P}, \mathcal{R}_\Box, \mathcal{R}_\Diamond)$ are $I$-compatible.
\end{lemma}
\begin{proof}
  We prove that $ R_\Box^{(0)}[\{\alpha/y\}]$ is Galois-stable for every $y$ occurring in $\overline{\mathcal{A}}$,  and any $\alpha \in \mathbf{H}$. The proofs for other conditions are similar. We consider two cases: whether $\Box y\in\overline{\mathcal{A}}$ or not. 

When $\Box y\in\overline{\mathcal{A}}$, for any $\beta\in \mathbf{H}$ and $b\in A$, if $\beta\leq  I(b,\Box y)$ (resp.~$\beta \leq  R_\Box (b,y))$ occurs in $\overline{\mathcal{A}}$, then $\beta \leq  R_\Box (b,y)$ (resp.~$ \beta \leq  I (b,\Box y)$) occurs in $\overline{\mathcal{A}}$ by $\Box y$ (resp.~$R_\Box$) rule. Therefore, by construction $I(b, \Box y)=R_\Box(b, y)$ for any $b\in A$, which implies $ R_\Box^{(0)}[\{\alpha/y\}](b) = I^{(0)}[\{\alpha/\Box y\}](b)$ for any $b\in A$ and $\alpha \in \mathbf{H}$. Therefore, for any $\alpha \in \mathbf{H}$, $ R_\Box^{(0)}[\{\alpha/y\}] = I^{(0)}[\{\alpha/\Box y\}]$, and hence $ R_\Box^{(0)}[\{\alpha/y\}]$ is Galois-stable. When $\Box y$ does not occur in $\overline{\mathcal{A}}$, no term of form  $R_\Box (b, y)$ occurs in $\overline{\mathcal{A}}$. In such case, we have for any $\alpha$, $ R_\Box^{(0)}[\{\alpha/y\}] =I^{(0)}[\{1/x_\bot\}] =\emptyset$ is Galois-stable.  
\end{proof}
From the  above lemmas, it immediately follows that the tuple $\mathcal{M} =(\mathfrak{F}, \cdot^\mathcal{M})$ defined at the beginning of the present section is an interpretation for $\mathrm{LE}$-$\mathcal{FALC}$. The following lemma states that the interpretation of any concept $C$ in $\mathcal{M}$ is completely determined by the ABox terms of the form $\alpha \leq  I (b,x_C)$ and $\alpha \leq I (a_C,y)$ occurring in $\overline{\mathcal{A}}$.

\begin{lemma} \label{lem:Soundness pre}
Let  $\mathcal{M}=(\mathfrak{F}, \cdot^\mathcal{M})$ be the interpretation defined by the construction
above.  Then for any concept $C$ and individuals $b$, $y$ that occur in ${\mathcal{A}}$, we have $\val{C}(b)=\max\{ \alpha \mid \alpha \leq  I (b,x_C) \in  \overline{\mathcal{A}}\}$, and $\descr{C}(y) = \max \{ \alpha \mid \alpha \leq  I (a_C ,y) \in  \overline{\mathcal{A}}\}$.
\end{lemma}

\begin{proof}
    See Appendix \ref{Proof of lem:Soundness pre}.
\end{proof}

\begin{theorem}[Soundness]\label{thm:Soundness}
Let $\mathcal{A}$ be an $\mathrm{LE}$-$\mathcal{FALC}$ ABox and $\overline{\mathcal{A}}$ be its completion which does not contain any clash, then the interpretation $\mathcal{M}=(\mathfrak{F}, \cdot^\mathcal{M})$ defined above is a model for $\mathcal{A}$. 
\end{theorem}
\begin{proof}
We show that $\mathcal{M}$ is a model for $\overline{\mathcal{A}}$, so it is a model for $\mathcal{A}$. The proof is  by cases.

\noindent 1.~By construction, $\mathcal{M}$ satisfies all the ABox assertions of the form $\alpha \leq t$ in $\overline{\mathcal{A}}$, where $t$ is an ABox term of the form $I(a,x)$, $R_\Box(a,x)$ or $R_\Diamond(a,x)$. 

\noindent 2.~Let $\alpha \not \leq t$ be an ABox assertion in $\overline{\mathcal{A}}$, where $t$ is a relational term. By construction, $\mathcal{M}\not\vDash \alpha \not\leq t$ iff  $\alpha' \leq t \in \overline{\mathcal{A}}$ for some $\alpha \leq \alpha'$. However, if such ABox assertion $\alpha' \leq t$ occurs in $\overline{\mathcal{A}}$, it clashes with the  ABox assertion $\alpha \not \leq t$. Therefore, as $\overline{\mathcal{A}}$ contains no clash, we have $\mathcal{M}\vDash \alpha \not \leq t$.  

\noindent 3.~For the ABox assertion of the form $\alpha \leq b:C$ (resp.~$\alpha \leq  y::C$)  occurring in $\overline{\mathcal{A}}$, by expansion rules the ABox assertion of the form  $\alpha\leq I(b, x_C)$ (resp.~$\alpha \leq I(a_C, y)$) occurs in $\overline{\mathcal{A}}$. By Lemma \ref{lem:Soundness pre}, in the interpretation $\mathcal{M}$, $\alpha\leq\val{C}(b)$ (resp.~$\alpha\leq\descr{C}(y)$), which means that $\mathcal{M}\vDash\alpha\leq b:C$ (resp.~$\mathcal{M}\vDash\alpha\leq y::C$).

\noindent 4.~For ABox assertions of the form $\alpha \not\leq b:C$ (resp.~$\alpha\not\leq y::C$) occurring in $\overline{\mathcal{A}}$, by Lemma \ref{lem:Soundness pre}, $\mathcal{M}\not\vDash\alpha \not\leq b:C$ (resp.~$\mathcal{M}\not\vDash\alpha \not\leq  y::C$) iff some term of the form $\alpha' \leq  b:C$ (resp.~ $\alpha' \leq  y::C$) occurs in $\overline{\mathcal{A}}$ for some $\alpha \leq  \alpha'$, which implies $\overline{\mathcal{A}}$ contains a clash. Therefore, $\mathcal{M} \vDash\alpha \not\leq b:C$ (resp.~$\mathcal{M} \vDash\alpha \not\leq  y::C$) because $\overline{\mathcal{A}}$ contains no clash.
\end{proof}

The following corollary is an immediate consequence of the termination and soundness of the tableaux algorithm.
\begin{corollary}[Finite Model Property]
For any consistent $\mathrm{LE}$-$\mathcal{FALC}$ ABox $\mathcal{A}$, there exists a model of $\mathcal{A}$ of size polynomial in $size(\mathcal{A})$.
\end{corollary}
\begin{proof}
The interpretation $\mathcal{M}$ defined above is the required model. The polynomial bound on the size of $\mathcal{M}$ follows from the proof of Theorem \ref{thm:termination}.
\end{proof}

\section{Completeness of the tableaux algorithm}
\label{sec:Completeness}
In this section, we  prove the completeness of the tableaux algorithm given in Section \ref{Sec: tableau}. The following lemma is key to this end, since it shows that  every model for an $\mathrm{LE}$-$\mathcal{FALC}$ ABox can be extended to a model with classifying objects and features.

\begin{lemma}\label{lem:characteristic consistency} 

For any ABox $\mathcal{A}$, any model $\mathcal{M}=(\mathfrak{F}, \cdot^{\mathcal{M}})$ of $\mathcal{A}$ can be extended to a model $\mathcal{M}'=(\mathfrak{F}', \cdot^{\mathcal{M}'})$ of $\mathcal{A}$ such that $\mathfrak{F}'=(A',X',I',\{R_\Box'\}_{\Box\in\mathcal{G}}, \{R_\Diamond'\}_{\Diamond\in\mathcal{F}})$, $A \subseteq A'$ and $X \subseteq  X'$, and moreover for every $\Box \in \mathcal{G}$ and $\Diamond \in \mathcal{F}$:

1.~For any concept $C$, there exists $a_C \in A'$ and $x_C \in X'$  such that:
\begin{equation}\label{eq:completness 1}
   C^{\mathcal{M}'} =(I'^{(0)}[\{1/x_C\}], I'^{(1)}[\{1/a_C\}]), \quad \val{C^{\mathcal{M}'}}(a_C)=1, \quad   \descr{C^{\mathcal{M}'}}(x_C)=1,  
\end{equation}

 2.~For every individual $b \in A$, and $\alpha\in\mathbf{H}$, there exist $\Diamond b, \Diamondblack b \in A'$ such that:
\begin{equation}\label{eq:completness 2}
I'^{(1)}[\{\alpha/ \Diamondblack b\}] = R_\Box'^{(1)}[\{\alpha/ b^{\mathcal{M}'}\}] \quad \mbox{and} \quad I'^{(1)}[\{\alpha/\Diamond b\}] = R_\Diamond'^{(0)}[\{\alpha/b^{\mathcal{M}'} \}],
\end{equation}

3.~For every individual $y \in X$, and $\alpha\in\mathbf{H}$, there exist $\Box y, \blacksquare y \in X'$ such that:
\begin{equation}\label{eq:completness 3}
I'^{(0)}[\{ \alpha/ \blacksquare y\}] = R_\Diamond'^{(1)}[\{\alpha/ y^{\mathcal{M}'}\}] \quad \mbox{and} \quad I'^{(0)}[\{ \alpha/ \Box y\}] = R_\Box'^{(0)}[\{ \alpha/ y^{\mathcal{M}'}\}]. 
\end{equation}

4.~For any concept $C$, and any $a \in A$, $x \in X$, 
\begin{equation}\label{eq:completeness 4}
\val{C^\mathcal{M}}(a)= \val{C^{\mathcal{M}'}} (a) \quad \mbox{and} \quad  \descr{C^\mathcal{M}}(x)= \descr{C^{\mathcal{M}'}}(x).    
\end{equation}
\end{lemma}

\begin{proof}
Fix $\Box \in \mathcal{G}$ and $\Diamond \in \mathcal{F}$. Let $\mathcal{M}'=(\mathfrak{F}', \cdot^{\mathcal{M}'})$, where $\mathfrak{F}'=(A',X',I',\{R_\Box'\}_{\Box\in\mathcal{G}}, \{R_\Diamond'\}_{\Diamond\in\mathcal{F}})$, be defined as follows: For every concept $C$,  we add new elements $a_C$ and $x_C$ to $A$ and $X$ (respectively) to obtain the sets $A'$ and $X'$. 
For any  $J \in \{I, R_\Box\}$, $a \in A'$ and $x \in X'$, the value of $J'(a, x)$ is defined as follows:

\noindent 1.~If $a \in A$, $x \in X$, then  $J'(a, x)=J(a,x)$; 
    
\noindent 2.~If $x \in X$, and $a=a_C$ for some concept $C$, then  $J'(a,x)= \bigwedge_{b\in A} (\val{C^{\mathcal{M}}}(b) \rightarrow J (b,x)) $;
    
\noindent 3.~If $a \in A$, and $x=x_C$ for some concept $C$, then  $J (a,x)= \bigwedge_{y\in X} (\descr{C^{\mathcal{M}}}(y) \rightarrow J (a, y)) $;
    
\noindent 4.~If $a=a_{C_1}$ and $x=x_{C_2}$ for some concepts $C_1$, $C_2$, then   $J'(a,x)=\bigwedge_{b \in A} \bigwedge_{y \in X} ((\val{C_1^{\mathcal{M}}}(b)\wedge  \descr{C^{\mathcal{M}}}(y)) \rightarrow J(b,y))$.
    
For any $a \in A'$ and $x \in X'$, the value of $R_\Diamond'(a, x)$ is defined as follows:

\noindent 1.~If $a\in A$, $x \in X$, then  $R_\Diamond'(x,a)=R_\Diamond(x,a)$; 
    
\noindent 2.~If $x\in X$, and $a=a_C$ for some concept $C$, then  $R_\Diamond (x,a)= \bigwedge_{b\in A} (\val{C^{\mathcal{M}}}(b) \rightarrow R_\Diamond (x, b)) $;
    
\noindent 3.~If $a \in A$, and $x=x_C$ for some concept $C$, then  $R_\Diamond (x,a)= \bigwedge_{y\in X} (\descr{C^{\mathcal{M}}}(y) \rightarrow R_\Diamond (y,a)) $;
    
\noindent 4.~If $a=a_{C_1}$ and $x=x_{C_2}$ for some concepts $C_1$, $C_2$, then   $R_\Diamond (x,a)=\bigwedge_{b\in A} \bigwedge_{y \in X} ((\val{C_1^{\mathcal{M}}}(b) \wedge\descr{C^{\mathcal{M}}}(y)) \rightarrow R_\Diamond (y,b))$.

For any $b \in A$, $y \in X$, let $\Diamondblack b =a_{\Diamondblack(cl(b))}$, $\Diamond b =a_{\Diamond(cl(b))}$,  $\blacksquare y=x_{\blacksquare (cl(y))}$, and   $\Box  y=x_{\Box (cl(y))}$, where $cl(b)$ (resp.~$cl(y)$) is the smallest fuzzy formal concept generated by $\{1/b\}$ (resp.~$\{1/y\}$), and the operations $\Diamondblack$ and $\blacksquare$ are the adjoints of the operations $\Box$, and $\Diamond$, respectively. For any concept $C$, let $C^{\mathcal{M}'} =(I'^{(0)}[\{1/x_C\}], I'^{(1)}[\{1/a_C\}])$. 

It is straightforward to check that $\mathcal{M}'$ is a model for $\mathcal{A}$ and satisfies all the properties required in Lemma \ref{lem:characteristic consistency}. This concludes the proof.
\end{proof}

\begin{theorem}[Completeness]\label{thm:completeness}
    Let $\mathcal{A}$ be a consistent ABox and $\mathcal{A}'$ be obtained via the application of any expansion rule  applied to $\mathcal{A}$, then $\mathcal{A}'$ is also consistent. 
\end{theorem}
\begin{proof}
If $\mathcal{A}$ is consistent, by  Lemma \ref{lem:characteristic consistency},  there exists a model $\mathcal{M}'$ of $\mathcal{A}$ which satisfies \eqref{eq:completness 1}, \eqref{eq:completness 2} and \eqref{eq:completness 3}. The theorem follows from the fact that any ABox assertion added by any expansion rule is satisfied by $\mathcal{M}'$, where we interpret $a_C$, $x_C$, $\Diamondblack b$, $\Diamond b$, $\Box y$, $\blacksquare y$ as in Lemma \ref{lem:characteristic consistency}. 
\end{proof}

As a demonstration of the functioning of the tableaux algorithm, we use tableaux Algorithm \ref{alg:main algo} (using unraveling to deal with TBox axioms) to show that  the  knowledge base $\mathcal{K}$   discussed in Section \ref{sec:example of knowledge bases} is consistent and construct a model for it  (cf.~Appendix \ref{app:model for example}),   while  the   knowledge base $\mathcal{K}'$  discussed in Section \ref{sec:example of knowledge bases}  is inconsistent  (cf.~Appendix \ref{App:clash example} for the proof).

\section{Conclusions and future directions}
In this paper, we  define a fuzzy lattice-based non-distributive description logic $\mathrm{LE}$-$\mathcal{FALC}$ as a description logic to describe and reason about fuzzy formal concepts arising from fuzzy (enriched) formal contexts, and define a polynomial-time tableaux algorithm to check the consistency of $\mathrm{LE}$-$\mathcal{FALC}$ ABoxes. Additionally, this algorithm  can be extended to an exponential-time algorithm for checking consistency of knowledge bases with acyclic TBoxes. This work can be extended in several interesting directions, including, but not limited to, the following.

\fakeparagraph{Dealing with cyclic TBoxes and  RBox axioms.} 
In this paper, we introduced a tableaux algorithm only for knowledge bases with acyclic TBoxes. In the future, we intend to generalize the algorithm to deal with cyclic TBoxes as well. Another interesting avenue of research is to  develop  tableaux algorithms for extensions of $\mathrm{LE}$-$\mathcal{FALC}$ with RBox axioms. RBox axioms are used in  description logics to describe the relationship between different relations in  knowledge bases and  the properties of these relations such as reflexivity, symmetry, and transitivity.  It would be interesting to see if it is possible to obtain necessary and/or sufficient conditions on the shape of RBox axioms for which a tableaux algorithm can be obtained. This has an interesting relationship with the problem in LE-logic of providing computationally efficient proof systems for various extensions of LE-logic in a modular manner \cite{greco2016unified,ICLArough}.

\fakeparagraph{Generalizing to more expressive description logics.} The  DL LE-$\mathcal{ALC}$ is the non-distributive counterpart  of  $\mathcal{ALC}$. A natural direction for further research is to explore the non-distributive counterparts of extensions of $\mathcal{ALC}$ such as  $\mathcal{ALCI}$ and $\mathcal{ALCIN}$ and fuzzy generalizations of such description logics. 
This would allow us to express more constructions like (fuzzy) concepts generated by an object or a feature,  which can not be expressed in ($\mathrm{LE}$-$\mathcal{FALC}$) $\mathrm{LE}$-$\mathcal{ALC}$.

 \fakeparagraph{Description logic and Formal Concept Analysis.} 
 The relationship between FCA and DL has been studied and used in several applications \cite{DLandFCA1,DLandFCA2,DLandFCA3}. The framework of LE-$\mathcal{FALC}$ formally brings Fuzzy FCA and DL together, both because its concepts are naturally interpreted as formal concepts in $\mathbf{H}$-valued FCA, and because its language is designed to represent knowledge and reasoning in enriched formal contexts. 
 Thus, these results can help integrate FCA and DL further in both  theory and applications.

\bibliographystyle{eptcs}
\bibliography{generic}

\appendix
\section{Proofs}\label{app:Proofs}

In this appendix, we gather proofs of some results stated throughout the paper.

\subsection{Proof of lemma \ref{lem:Soundness pre}}\label{Proof of lem:Soundness pre}

\begin{proof}
Note that by rule $MV_\vee$, the maximums in the statement of the lemma exist and are unique when the algorithm terminates. The proof is by induction on the complexity of concept $C$. The base case (when $C$ is a primitive concept) is immediate by  construction of the model. For the induction step, we have four cases.

\noindent 1.~Suppose $C=C_1 \wedge C_2$. 

\noindent 1.1.~For the first claim, it is required to prove that
\begin{equation*}
\max \{ \alpha \mid \alpha\leq I(b, x_{C_1 \wedge C_2})\in\overline{\mathcal{A}}\}=\val{C_1 \wedge C_2}(b).
\end{equation*}

For direction $(\leq)$, suppose $\alpha_0 \leq b I x_{C_1 \wedge C_2} \in \overline{\mathcal{A}}$ for some $\alpha_0$. Then by using appending and $\wedge_A$ rules, we have $\alpha_0\leq I(b,x_{C_1})\in\overline{\mathcal{A}}$ and $\alpha_0 \leq I(b, x_{C_2})\in \overline{\mathcal{A}}$. Therefore, $\alpha_0\leq\max \{ \alpha \mid \alpha \leq  I (b, x_{C_1}) \}\in  \overline{\mathcal{A}}\}$, and $\alpha_0\leq\max \{ \alpha \mid \alpha \leq  I (b, x_{C_2}) \}\in  \overline{\mathcal{A}}\}$. Hence, by introduction hypothesis, we have $\alpha_0\leq\val{C_1}(b)$ and $\alpha_0\leq\val{C_2}(b)$, which implies $\alpha_0\leq\val{C_1}(b)\wedge\val{C_2}(b)=\val{C_1 \wedge C_2}(b)$. Therefore, we have $\max\{\alpha\mid\alpha\leq I(b, x_{C_1 \wedge C_2})\in\overline{\mathcal{A}}\}\leq\val{C_1 \wedge C_2}(b)$.

For direction $(\geq)$, suppose $\alpha_0 \leq \val{C_1 \wedge C_2}(b)=\val{C_1}(b)\wedge\val{C_2}(b)$. Therefore, we have $\alpha_0 \leq \val{C_1} (b)$, and $\alpha_0 \leq \val{C_2} (b)$. Hence, by induction hypothesis, there are $\alpha_1, \alpha_2 \in \mathbf{H}$, such that $\alpha_0\leq\alpha_1$, $\alpha_0 \leq \alpha_2$, $\alpha_1 \leq I(b, x_{C_1})$, and $\alpha_2 \leq I(b, x_{C_2}) \in \overline{\mathcal{A}}$. As $C_1 \wedge C_2$ occurs in $\overline{\mathcal{A}}$, by appending and $\wedge_A^{-1}$ rules, we get  $\alpha_1 \wedge \alpha_2 \leq b :C_1 \wedge C_2\in\overline{\mathcal{A}}$. By creation and the basic rules, this implies $\alpha_1 \wedge \alpha_2  \leq b I x_{C_1 \wedge C_2} \in \overline{\mathcal{A}}$. Therefore, we have $\alpha_0 \leq \max \{ \alpha \mid \alpha \leq I (b, x_{C_1 \wedge C_2})\in  \overline{\mathcal{A}}\}$. Hence,  by completeness of $\mathbf{H}$, we have $ \val{C_1 \wedge C_2} (b) \leq \max \{ \alpha \mid \alpha \leq I (b, x_{C_1 \wedge C_2})\in\overline{\mathcal{A}}\}$.

\noindent 1.2.~For the second claim, it is required to prove that
\begin{equation*}
    \descr{C_1 \wedge C_2}(y)=\max \{ \beta \mid \beta \leq  I (a_{C_1 \wedge C_2}, y)\in  \overline{\mathcal{A}}\}.
\end{equation*}

For direction $(\leq)$, notice that $\descr{C_1 \wedge C_2}(y)=\wedge_{b \in A} (\val{C_1\wedge C_2}(b)\rightarrow I(b,y))$. By the proof of the first claim above and the construction of the model, we have $ \descr{C_1 \wedge C_2}(y) =  \wedge_{b \in A} (  \max \{ \alpha \mid \alpha \leq  I (b, x_{C_1 \wedge C_2})\in  \overline{\mathcal{A}}\} \rightarrow \max \{ \beta \mid \beta \leq  I (b, y)\in  \overline{\mathcal{A}}\})$. Suppose $\alpha_0 \leq \descr{C_1 \wedge C_2}(y)$. Therefore, we have $\alpha_0 \leq \max \{ \alpha \mid \alpha \leq  I (b, x_{C_1 \wedge C_2})\in  \overline{\mathcal{A}}\} \rightarrow \max \{ \beta \mid \beta \leq I(b,y)\in\overline{\mathcal{A}}\}$ for every $b\in A$. By the creation rule, we have $1 \leq I(a_{C_1 \wedge C_2}, x_{C_1 \wedge C_2}) \in  \overline{\mathcal{A}}$, which means that $\max\{ \alpha \mid \alpha \leq  I (b, x_{C_1\wedge C_2})\in\overline{\mathcal{A}}\}=1$. Therefore, we get $\alpha_0 \leq  \max \{ \beta \mid \beta \leq  I (a_{C_1 \wedge C_2}, y) \in  \overline{\mathcal{A}}\}$. Hence,
by completeness of  $\mathbf{H}$,
we have $\descr{C_1 \wedge C_2}(y) \leq \max \{ \beta \mid \beta \leq I(a_{C_1 \wedge C_2}, y)\in  \overline{\mathcal{A}}\}$. 

For direction $(\geq)$, suppose $\alpha_0 \leq I (a_{C_1 \wedge C_2}, y)$. By construction of the model, there exists $\alpha_1 \in \mathbf{H}$, such that $\alpha_0 \leq \alpha_1$ and $\alpha_1 \leq I (a_{C_1 \wedge C_2}, y) \in  \overline{\mathcal{A}}$. Suppose $\alpha_2 \leq  I (b, x_{C_1 \wedge C_2}) \in \overline{\mathcal{A}}$ for some $b \in A$, and $\alpha_2 \in \mathbf{H}$, by appending and basic rules, we get $\alpha_1 \wedge \alpha_2 \leq I(b,y) \in  \overline{\mathcal{A}}$. Hence, we have $\alpha_1 \wedge\alpha_2\leq\max\{ \alpha \mid \alpha \leq  I (b, y )\in\overline{\mathcal{A}}\}$ for any $\alpha_2 \leq  I (b, x_{C_1 \wedge C_2})\in\overline{\mathcal{A}}$. Therefore, by properties of $\mathbf{H}$, we have $\alpha_0 \leq \alpha_1 \leq \alpha_2 \rightarrow \alpha_1 \wedge \alpha_2 \leq \max \{ \alpha \mid \alpha \leq  I (b, x_{C_1 \wedge C_2})\in  \overline{\mathcal{A}}\} \rightarrow \max \{ \beta \mid \beta \leq  I (b, y)\in  \overline{\mathcal{A}}\}$. Hence, by completeness of $\mathbf{H}$, we have $ I (a_{C_1 \wedge C_2}, y) \leq \max \{ \alpha \mid \alpha \leq  I (b, x_{C_1 \wedge C_2})\in  \overline{\mathcal{A}}\} \rightarrow \max \{ \beta \mid \beta \leq  I (b, y)\in\overline{\mathcal{A}}\}$ for any $b\in A$. Therefore, we have $\max\{\beta \mid\beta\leq I (a_{C_1 \wedge C_2}, y)\}\leq\wedge_{b \in A} (\val{C_1\wedge C_2}(b)\rightarrow I(b,y))=\descr{C_1 \wedge C_2}(y)$.

\noindent 2.~The proof for  $C=C_1 \vee C_2$  is similar to the previous one.

\noindent 3.~Suppose $C=[R_\Box] C_1$.

\noindent 3.1.~For the first claim, it is required to prove that 
\begin{equation*}
\val{[R_\Box] C_1}(b)=\max\{ \beta  \mid \beta  \leq I (b, x_{\Box C_1})\in \overline{\mathcal{A}} \}.
\end{equation*}

For direction $(\leq)$, notice that $\val{[R_\Box] C_1}(b)= \bigwedge_{y \in X} (\descr{C_1}(y) \rightarrow R_\Box (b,y))$. By induction and construction of the model, this is equivalent to $\val{[R_\Box] C_1}(b)= \bigwedge_{y \in X} (\max \{ \alpha \mid \alpha \leq  I (a_{C_1},y) \}  \rightarrow \max \{ \beta  \mid \beta  \leq R_\Box (b, y) \in \overline{\mathcal{A}} \})$. Suppose $\alpha_0 \leq \val{[R_\Box] C_1}(b)$. By creation and basic rules, we have $1 \leq I (a_{C_1},x_{C_1}) \in \overline{\mathcal{A}}$. By induction hypothesis claim above, we get $\alpha_0 \leq \max \{ \beta  \mid \beta  \leq  R_\Box (b, x_{C_1}) \in \overline{\mathcal{A}} \}$. By  rule $R_\Box$, this implies $\alpha_0 \leq \max \{ \beta  \mid \beta  \leq  I (b, \Box x_{C_1})=I(b, x_{\Box C_1})\in \overline{\mathcal{A}} \}$.  Therefore, $\val{[R_\Box] C_1}(b) \leq \max \{ \beta\mid \beta\leq I(b, x_{\Box C_1})\in\overline{\mathcal{A}}\}$.  

For direction $(\geq)$, suppose $\alpha_0  \leq I(b, x_{\Box C_1})$. By construction of the model, we have $\alpha_1\leq I(b, x_{\Box C_1})\in\overline{\mathcal{A}}$ for some $\alpha_0\leq\alpha_1\in \mathbf{H}$. Suppose $\alpha_2 \leq  I (a_{C_1},y) \in \overline{\mathcal{A}}$ for some $y$. Then by appending and $\Box$ rules, we get $\alpha_1 \wedge \alpha_2  \leq R_\Box (b, y)$. By properties of $\mathbf{H}$, we have $\alpha_0  \leq   \alpha_1 \leq \alpha_2 \rightarrow \alpha_1 \wedge \alpha_2$. Hence we get $\alpha_0 \leq \max \{ \alpha \mid \alpha \leq  I (a_{C_1},y) \}  \rightarrow \max \{ \beta  \mid \beta  \leq  R_\Box (b, y) \in \overline{\mathcal{A}} \}$ for any $y$. Therefore, $\max \{ \beta  \mid \beta  \leq   I (b, x_{\Box C_1})\in \overline{\mathcal{A}} \} \leq \val{[R_\Box] C_1}(b)$.

\noindent 3.2.~For the second claim, it is required to prove that 
\begin{equation*}
\descr{[R_\Box C_1]}(y)=\max\{ \beta \mid \beta \leq  I (a_{\Box C_1}, y)\in\overline{\mathcal{A}}\}.
\end{equation*}

For direction $(\leq)$, notice that $\descr{[R_\Box] C_1 }(y)=\wedge_{b \in A}(\val{[R_\Box] C_1}(b)\rightarrow I(b,y))$. By the proof of the first claim and construction of the model, we have 
$ \descr{[R_\Box C_1]}(y) =  \wedge_{b \in A} (  \max \{ \alpha \mid \alpha \leq  I (b, x_{\Box C_1})\} \rightarrow \max \{ \beta \mid \beta \leq  I (b, y)\})$. Suppose $\alpha_0 \leq \descr{[R_\Box] C_1}(y)$. Then for every $b$, $\alpha_0 \leq \max \{ \alpha \mid \alpha \leq  I (b, x_{\Box C_1})\} \rightarrow \max \{ \beta \mid \beta \leq  I (b, y)\}$. By applying this to $1 \leq I(a_{\Box C_1}, x_{\Box C_1}) \in  \overline{\mathcal{A}}$ added by creation rule, we get $\alpha_0 \leq  \max \{ \beta \mid \beta \leq  I (a_{\Box C_1}, y)\}$. Therefore, $ \descr{[R_\Box C_1]}(y) \leq \max \{ \beta \mid \beta \leq  I (a_{\Box C_1}, y)\in\overline{\mathcal{A}}\}$.

For direction $(\geq)$, suppose $\alpha_0 \leq I (a_{\Box C_1 }, y)$. By construction of the model, there exists $\alpha_1 \in \mathbf{H}$ such that $\alpha_0 \leq \alpha_1$ and $\alpha_1 \leq I (a_{\Box C_1}, y) \in  \overline{\mathcal{A}}$.  Suppose $\alpha_2 \leq  I (b, x_{\Box C_1}) \in \overline{\mathcal{A}}$ for some $b \in A$, $\alpha_2 \in \mathbf{H}$. Then, by appending and basic rule, we get $\alpha_1 \wedge \alpha_2 \leq I(b,y) \in  \overline{\mathcal{A}}$. Therefore, for   for any $\alpha_2 \leq  I (b, x_{\Box C_1 }) \in \overline{\mathcal{A}}$, we have 
$\alpha_1 \wedge \alpha_2 \leq  \max \{ \alpha \mid \alpha \leq  I (b, y )\} \in  \overline{\mathcal{A}} $.  Hence, we have $\alpha_0 \leq \alpha_1 \leq \alpha_2 \rightarrow \alpha_1 \wedge \alpha_2 \leq \max \{ \alpha \mid \alpha \leq  I (b, x_{\Box C_1})\in  \overline{\mathcal{A}}\} \rightarrow \max \{ \beta \mid \beta \leq  I (b, y)\in  \overline{\mathcal{A}}\}$. Therefore, by the completeness of $\mathbf{H}$, we have $ I (a_{\Box C_1}, y) \leq \max \{\alpha \mid \alpha \leq  I (b, x_{\Box C_1})\in\overline{\mathcal{A}}\} \rightarrow \max \{ \beta \mid \beta \leq  I (b, y)\in  \overline{\mathcal{A}}\}$. Therefore, we have $\max\{ \beta \mid \beta \leq  I (a_{\Box C_1}, y)\in\overline{\mathcal{A}}\}\leq\descr{[R_\Box C_1]}(y)$.

\noindent 4. The proof for  $C=\langle R_\Diamond \rangle C_1$  is similar to the previous one. 

This concludes the proof.
\end{proof}

\section{Model for the first example knowledge base}\label{app:model for example}

In this section, we describe model for the  first knowledge base defined in Section \ref{sec:example of knowledge bases} obtained using unraveling and Tableaux Algorithm \ref{alg:main algo}.

Let $\mathcal{K} = (\mathcal{A}, \mathcal{T})$ be a knowledge base such that  $\mathcal{T} =
\{ C_4  \equiv C_1 \vee  C_3,   C_2 \equiv   C_4 \wedge C_7, C_6 \equiv (C_1 \wedge C_2) \vee (C_2 \wedge C_3) 
\}$,  and $\mathcal{A}=\{1 \leq P_1:C_2, P_1:C_6< 1/2, 1 \leq y_1::[R_\Box] C_1,  1 \leq y_3::[R_\Box] C_3, P_2 R_\Box y_3<1,  1/2\leq P_2:C_1, 1/2 \leq y_1::C_1  \}$.  By unraveling TBox we get the following concept definitions.
\begin{enumerate}
    \item $C_4  \equiv C_1 \vee  C_3$
    \item $C_2 \equiv (C_1 \vee  C_3)   \wedge C_7$
    \item $C_6 \equiv (C_1 \wedge((C_1 \vee  C_3)   \wedge C_7) ) \vee (((C_1 \vee  C_3)   \wedge C_7)\wedge C_3) \equiv (C_1 \wedge C_7) \vee (C_7 \wedge C_3) $.
\end{enumerate}
By substituting these definitions in ABox and running Algorithm \ref{alg:main algo}, we get a  model for $\mathcal{K}$ using construction described in Section \ref{sec:soundness}.  Table \ref{tab:constants} lists symbols we use for different constants of the form $a_C$ or $x_C$ which appear in the model. 

\begin{table}[h]
\centering
 \begin{tabular}{|c|c|c|c|}
\hline
   $a_1$  & $a_{C_1}$ &    $x_1$  & $x_{C_1}$   \\
    $a_2$  & $a_{C_3}$ &    $x_2$  & $x_{C_3}$   \\
    $a_3$  & $a_{C_1 \vee C_3}$ &    $x_3$  & $x_{C_1 \vee C_3}$  \\
 $a_4$  & $a_{(C_1 \vee C_3) \wedge C_7}$ &    $x_4$  & $x_{(C_1 \vee C_3) \wedge C_7}$\\
     $a_5$  & $a_{C_7}$ &    $x_5$  & $x_{C_7}$   \\
  $a_6$  & $a_{C_1 \wedge C_7}$ &    $x_6$  & $x_{C_1 \wedge C_7}$ \\
   $a_7$  & $a_{C_3 \wedge C_7}$ &    $x_7$  & $x_{C_3 \wedge C_7}$ \\
    $a_8$  & $a_{(C_1 \wedge C_7) \vee (C_3 \wedge C_7)}$ &    $x_8$  & $x_{(C_1 \wedge C_7) \vee (C_3 \wedge C_7)}$\\
     $a_{9}$  & $a_{\Box C_1}$ &       $x_{9}$  & $x_{\Box C_1}=\Box x_{C_1}$  \\
        $a_{10}$  & $a_{\Box C_3}$ &       $x_{10}$  & $x_{\Box C_3}=\Box x_{C_3}$  \\
         $a_{11}$  & $a_{\top}$ &    $x_{11}$  & $x_{\bot}$ \\
         \hline
\end{tabular}
\caption{symbols for constants of the form $a_C$ or $x_C$  appearing in the model}
    \label{tab:constants}
\end{table}

 Table \ref{tab:I relation} lists all the object and feature  constant names appearing in the model and value of $I$ relation for every object-feature pair. 

\begin{table}[h]
    \centering
    \begin{tabular}{|c|cccccccccccccccc|}
\hline 
& $x_1$ & $x_2$ & $x_3$ & $x_4$ & $x_5$ & $x_6$ & $x_7$ & $x_8$ & $x_9$ & $x_{10}$ & $x_{11}$ & $y_1$ & $y_3$ & $\Box y_1$&  $\Box y_3$& $\Box x_3$\\
\hline 
 $a_1$ & 1 & 0 &  1 & 0 & 0 & 0& 0& 0& 0& 0& 0 & 1/2 &0 &0 & 0 &0 \\
 $a_2$ & 0 & 1 &  1 & 0 & 0 & 0& 0& 0& 0& 0& 0 & 0 &0 &0 & 0 &0\\
  $a_3$    & 0 & 0 &  1 & 0 & 0 & 0& 0& 0& 0& 0& 0 & 0 &0 &0 & 0  &0  \\
 $a_4$ & 0 & 0 & 0  & 1 & 0 & 0& 0& 0& 0& 0& 0 & 0 &0& 0 & 0 &0 \\
  $a_5$    & 0 & 0 &  1 & 1 & 1 & 0& 0& 0& 0& 0& 0 & 0 & 0& 0 & 0  &0   \\
 $a_6$ & 1 & 0 &  1 & 1 & 1 & 1& 0& 1& 0& 0& 0 & 1/2 & 0 & 0 & 0 &0 \\
  $a_7$    & 0 & 1 &  1 & 1 & 1 & 0& 1& 1& 0& 0& 0 & 0 &0& 0 & 0  &0   \\
 $a_8$ & 0 & 0 &  1 & 1 & 1 & 0& 0& 1& 0& 0& 0 & 0 &0& 0 & 0 &0 \\
  $a_9$    & 0 & 0 &  0 & 0 & 0 & 0& 0& 0& 1& 0& 0 & 1 &1/2& 0 & 0 &1     \\
 $a_{10}$ &  0 & 0 &  0 & 0 & 0 & 0& 0& 0& 0& 1& 0 & 0 &1 & 0 & 0&1\\
 $\Diamondblack a_{9}$    &   1 & 0 &  1 & 0 & 0 & 0& 0& 0& 0& 0& 0 & 1/2 &0& 0 & 0 &0  \\
$\Diamondblack a_{10}$    &  0 & 1 &  1 & 0 & 0 & 0& 0& 0& 0& 0& 0 & 0 & 0 & 0 & 0  &0  \\
  $a_{11}$    &   0 & 0 &  0 & 0 & 0 & 0& 0& 0& 0& 0& 0 & 0 &0& 0 & 0 &0  \\
 $P_1$  & 0 & 0 &  1 & 1 & 1 & 0& 0& 0& 0& 0& 0 & 0 &0& 0 & 0   &0\\
 $P_2$ & 1/2 & 0 &  1/2 & 1 & 1 & 0& 0& 0& 0& 0& 0 & 1/2 &0& 0 & 1/2  &0\\
 $\Diamondblack P_2$  &   0 & 0 &  0 & 0 & 0 & 0& 0& 0& 0& 0& 0 & 0 &0& 0 & 0 &0 \\ 
 \hline
\end{tabular}
      \caption{Objects ($A$) and features ($X$) of model and Relation $I$ between them }
    \label{tab:I relation}
\end{table}

The relation $R_\Box$ is defined by $R_\Box(a_9,x_1)=R_\Box(a_9,x_3)=R_\Box(a_{10},x_{2})=R_\Box(a_{10},x_{3})=0$, $R_\Box(a_9,y_1)=R_\Box(P_1,y_3)= 1/2$, and $R_\Box (b,y)=0$ for any other $b$, $y$. The relation $R_\Diamond$ is defined by $R_\Diamond (y,b)=0$ for any  $b$, $y$. The model contains atomic concepts $C_1$, $C_3$, and $C_7$. For any of these concepts $C$, its interpretation is given by  the tuple $(a_C^\uparrow, x_C^\downarrow)$. It is easy to verify that the relations $R_\Box$, and $R_\Diamond$ are $I$-compatible (In particular, it can be checked  that Lemma \ref{lem:Galois-stability} holds and the model defined above satisfies the knowledge base $\mathcal{K}$. 

\section{Proof of inconsistency of the second example knowledge base}\label{App:clash example}
In this section, we show that the second example of knowledge base given in Section \ref{sec:example of knowledge bases}, i.e.~knowledge base  $\mathcal{K}$  with TBox axiom $C_5 \equiv C_1 \wedge C_3$, and with set of ABox axioms $\{1/2 \leq P_3:[R_\Box] C_1, 1 \leq P_3:[R_\Box] C_3,  P_3 R_\Box y_4 <1/2,  1/2 \leq y_4:: C_5\}$ is inconsistent using tableaux algorithm. 

By unraveling, we replace any occurrence of $C_5$ in ABox with  $C_1 \wedge C_3$. The following table shows terms added to the tableaux expansion of  the resulting ABox. 

{{\centering
\begin{tabular}{ccc}
\hline
    \textbf{Rule} & \textbf{Premises} & \textbf{Added terms}\\
\hline
 create  &   & $1 \leq x_{C_1}::C_1$, $1 \leq x_{C_2}::C_2$\\ 
 $\Box $ & $1/2 \leq P_3:[R_\Box] C_1$, $1 \leq x_{C_1}::C_1$ & $1/2 \leq P_3 R_\Box x_{C_1}$ \\
  $\Box $ & $1 \leq P_3:[R_\Box] C_3$, $1 \leq x_{C_3}::C_3$ & $1 \leq P_3  R_\Box x_{C_3}$\\
  $R_\Box$ &  $1/2 \leq P_3 R_\Box x_{C_1}$, $1 \leq P_3  R_\Box x_{C_3}$ &  $1/2 \leq \Diamondblack P_3  Ix_{C_1}$, $1 \leq \Diamondblack P_3 I  x_{C_3}$\\
  $x_C$ & $1/2 \leq \Diamondblack P_3  Ix_{C_1}$, $1 \leq \Diamondblack P_3 I  x_{C_3}$ & $1/2 \leq  \Diamondblack P_3:C_3$,  $1 \leq  \Diamondblack P_3:C_1$\\
  $\wedge_A$ &  $1/2 \leq  \Diamondblack P_3:C_3$,  $1 \leq  \Diamondblack P_3:C_1$ & $1/2 \leq  \Diamondblack P_3:C_3 \wedge C_1$ \\
  create   & & $x_{C_1 \wedge C_3}:: C_1 \wedge C_3$ \\
  $I$ & $x_{C_1 \wedge C_3}:: C_1 \wedge C_3$, $1/2 \leq  \Diamondblack P_3:C_3 \wedge C_1$ & $1/2 \leq \Diamondblack P_3 I x_{C_1 \wedge C_3}$ \\
  $\Diamondblack b $ & $1/2 \leq \Diamondblack P_3 I x_{C_1 \wedge C_3}$ & $1/2 \leq  P_3 I \Box x_{C_1\wedge C_3}= 1/2 \leq  P_3 I  x_{\Box (C_1 \wedge C_3)}$ \\
  $x_C$ & $ 1/2 \leq  P_3 I  x_{\Box (C_1 \wedge C_3)}$ & $ 1/2 \leq  P_3 : [R_\Box] (C_1 \wedge C_3)$ \\
  unravel & $1/2 \leq y_4:: C_5$ & $1/2 \leq y_4:: C_1 \wedge C_3$ \\
  $\Box$ & $ 1/2 \leq  P_3 : [R_\Box] (C_1 \wedge C_3)$, $1/2 \leq y_4:: C_1 \wedge C_3$ & $1/2 \leq P_3 R_\Box y_4$\\
  \hline
\end{tabular}
\par}}
\smallskip
The term $1/2 \leq  P_3 R_\Box y_4$ clashes with the term  $P_3 R_\Box y_4 < 1/2$ in $\mathcal{A}$. Hence proved.
\end{document}